\documentclass[12pt]{article}

\usepackage{amsthm,amsmath,amssymb}
\usepackage[english]{babel}
\usepackage{graphicx}

\usepackage{hyperref}

\theoremstyle{plain}
\newtheorem{theorem}{Theorem}
\newtheorem{lemma}[theorem]{Lemma}
\newtheorem{corollary}[theorem]{Corollary}
\newtheorem{proposition}[theorem]{Proposition}

\theoremstyle{definition}
\newtheorem{definition}[theorem]{Definition}
\newtheorem{example}[theorem]{Example}

\theoremstyle{remark}
\newtheorem{remark}[theorem]{Remark}

\title{\bf A new approach to the $2$-regularity of the $\ell$-abelian complexity of $2$-automatic sequences}

\author{Aline Parreau\thanks{This work has been done when this author was an FNRS post-doctoral fellow at the University of Liege.}\\
\small LIRIS\\[-0.8ex]
\small University of Lyon, CNRS\\[-0.8ex] 
\small Lyon, France\\
\small\tt aline.parreau@univ-lyon1.fr\\
\and  
 Michel Rigo \\
\small Department of Mathematics\\[-0.8ex]
\small University of Liege\\[-0.8ex] 
\small Liege, Belgium\\
\small\tt M.Rigo@ulg.ac.be\\ 
\and
Eric Rowland\thanks{BeIPD-COFUND post-doctoral fellow at the University of Liege.}\\
 \small Department of Mathematics\\[-0.8ex]
\small University of Liege\\[-0.8ex] 
\small Liege, Belgium\\
\small\tt erowland@ulg.ac.be\\ 
\and
\'Elise Vandomme\\
\begin{tabular}{cc}
\small Department of Mathematics & \small Institut Fourier \\[-0.8ex]
\small University of Liege & \small University of Grenoble \\[-0.8ex] 
\small Liege, Belgium & \small Grenoble, France \\
\end{tabular} \\
\small\tt E.Vandomme@ulg.ac.be
}

\date{January 16, 2015}

\newcommand{\Z}{\mathbb{Z}}
\newcommand{\blo}{\textnormal{block}}
\newcommand{\fac}{\textnormal{Fac}}
\newcommand{\seq}[1]{\cite[\href{http://oeis.org/#1}{#1}]{OEIS}}
\DeclareMathOperator{\M}{max}	
\DeclareMathOperator{\m}{min}	
\DeclareMathOperator{\JM}{MJ}	
\DeclareMathOperator{\jm}{mj}	

\begin{document}

\maketitle

\begin{abstract}
We prove that a sequence satisfying a certain symmetry property is $2$-regular in the sense of Allouche and Shallit, i.e.,
    the $\mathbb{Z}$-module generated by its $2$-kernel is finitely
    generated.
We apply this theorem to develop a general approach for studying the $\ell$-abelian complexity of $2$-automatic sequences.
In particular, we prove that the period-doubling word and the Thue--Morse word have $2$-abelian complexity sequences that are $2$-regular.
Along the way, we also prove that the $2$-block codings of these two words have $1$-abelian complexity sequences that are $2$-regular.
\end{abstract}

\section{Introduction}

This paper is about some structural properties of integer sequences that occur naturally in combinatorics on words. Since the fundamental work of Cobham~\cite{Cobham 1972}, the so-called automatic sequences have been extensively studied. We refer the reader to \cite{Allouche--Shallit 2003b} for basic definitions and properties. These infinite words over a finite alphabet can be obtained by iterating a prolongable morphism of constant length to get an infinite word (and then, an extra letter-to-letter morphism, also called coding, may be applied). As a fundamental example, the {\em Thue--Morse word} $\mathbf{t}=\sigma^\omega(0)=0110100110010110\cdots$ is a fixed point of the morphism $\sigma$ over the free monoid $\{0,1\}^*$ defined by $\sigma(0)=01$, $\sigma(1)=10$. Similarly, the {\em period-doubling word} $\mathbf{p}=\psi^\omega(0)=01000101010001000100\cdots$ is a fixed point of the morphism $\psi$ over $\{0,1\}^*$ defined by $\psi(0)=01$, $\psi(1)=00$. We will discuss again these two examples of $2$-automatic sequences. 

Since an infinite word is just a sequence over $\mathbb{N}$ taking values in a finite alphabet, we use the terms `infinite word' and `sequence' interchangeably.

Let $k\ge 2$ be an integer. One characterization of $k$-automatic sequences is that their $k$-kernels are finite; see \cite{Eilenberg 1974} or \cite[Section~6.6]{Allouche--Shallit 2003b}. 
\begin{definition}
The {\em $k$-kernel} of a sequence $\mathbf{s}=s(n)_{n\ge 0}$ is the set $$\mathcal{K}_k(\mathbf{s})=\{s(k^in+j)_{n\ge 0} : i\ge 0 \text{ and } 0\le j<k^i\}.$$
\end{definition}
For instance, the $2$-kernel $\mathcal{K}_2(\mathbf{t})$ of the Thue--Morse word contains exactly two elements, namely $\mathbf{t}$ and $\sigma^\omega(1)$.

A natural generalization of automatic sequences to sequences on an infinite alphabet is given by the notion of $k$-regular sequences. We will restrict ourselves to sequences taking integer values only. 
\begin{definition}\label{def:regular}
Let $k\ge 2$ be an integer. A sequence $\mathbf{s}=s(n)_{n\ge 0}\in\mathbb{Z}^\mathbb{N}$ is {\em $k$-regular} if $\langle \mathcal{K}_k(\mathbf{s}) \rangle$ is a finitely-generated $\mathbb{Z}$-module, i.e., there exist a finite number of sequences $t_1(n)_{n\ge 0}, \ldots, t_\ell(n)_{n\ge 0}$ such that  
every sequence in the $k$-kernel $\mathcal{K}_k(\mathbf{s})$ is a $\mathbb{Z}$-linear combination of the $t_r$'s. Otherwise stated, for all $i\ge 0$  and for all $j\in\{0,\ldots,k^i-1\}$, there exist integers $c_1,\ldots,c_\ell$ such that
$$\forall n\ge 0,\quad s(k^in+j)=\sum_{r=1}^\ell c_r\, t_r(n).$$
\end{definition}
There are many natural examples of $k$-regular sequences~\cite{Allouche--Shallit 1992,Allouche--Shallit 2003}. There is a convenient matrix representation for $k$-regular sequences which leads to an efficient algorithm for computing the values of such a sequence (and many related quantities). See also \cite[Chapter~5]{Berstel--Reutenauer 2011} for connections with rational series. In particular, a sequence taking finitely many values is $k$-regular if and only if it is $k$-automatic. The $k$-regularity of a sequence provides us with structural information about how the different terms are related to each other.

A classical measure of complexity of an infinite word $\mathbf{x}$ is its {\em factor complexity} $\mathcal{P}^{(\infty)}_\mathbf{x}:\mathbb{N}\to\mathbb{N}$ which maps $n$ to the number of distinct factors of length $n$ occurring in $\mathbf{x}$. It is well known that a $k$-automatic sequence $\mathbf{x}$ has a $k$-regular factor complexity function and the sequence $(\mathcal{P}^{(\infty)}_\mathbf{x}(n+1)-\mathcal{P}^{(\infty)}_\mathbf{x}(n))_{n\ge 0}$ is $k$-automatic. See \cite{Carpi--DAlonzo 2010, Charlier--Rampersad--Shallit 2012} for a proof and relevant extensions. As an example, again for the Thue--Morse word, we have $$\mathcal{P}^{(\infty)}_\mathbf{t}(2n+1)=2\mathcal{P}^{(\infty)}_\mathbf{t}(n+1)\text{ and }\mathcal{P}^{(\infty)}_\mathbf{t}(2n)=\mathcal{P}^{(\infty)}_\mathbf{t}(n+1)+\mathcal{P}^{(\infty)}_\mathbf{t}(n)$$ for all $n\ge 2$. See also \cite{Frid} where a formula was obtained for the factor complexity of fixed points of some uniform morphisms.

Recently there has been a renewal of interest in abelian notions arising in combinatorics on words (e.g., avoiding abelian or $\ell$-abelian patterns, abelian bordered words, etc.). For instance, two finite words $u$ and $v$ are {\em abelian equivalent} if one is obtained by permuting the letters of the other one, i.e., the two words share the same Parikh vector, $\Psi(u)=\Psi(v)$. Since the Thue--Morse word is an infinite concatenation of factors $01$ and $10$, this word is {\em abelian periodic} of period $2$. The {\em abelian complexity}  of an infinite word $\mathbf{x}$ is a function $\mathcal{P}^{(1)}_\mathbf{x}:\mathbb{N}\to\mathbb{N}$ which maps $n$ to the number of distinct factors of length $n$ occurring in $\mathbf{x}$, counted up to abelian equivalence.  Madill and Rampersad~\cite{Madill--Rampersad} provided the first example of regularity in this setting: the abelian complexity of the paper-folding word (which is another typical example of an automatic sequence) is unbounded and $2$-regular.

Let $\ell\ge 1$ be an integer. Based on \cite{Karhumaki 1980} the notions of abelian equivalence and thus abelian complexity were recently extended to $\ell$-abelian equivalence and $\ell$-abelian complexity \cite{Karhumaki--Saarela--Zamboni}. 
\begin{definition}
Let $u,v$ be two finite words. We let $|u|_v$ denote the number of occurrences of the factor $v$ in $u$. Two finite words $x$ and $y$ are {\em $\ell$-abelian equivalent} if $|x|_v=|y|_v$ for all words $v$ of length $|v| \leq \ell$. 
\end{definition}
As an example, the words $011010011$ and  $001101101$ are $2$-abelian equivalent but not $3$-abelian equivalent (the factor $010$ occurs in the first word but not in the second one). 
Hence one can define the function $\mathcal{P}^{(\ell)}_\mathbf{x}:\mathbb{N}\to\mathbb{N}$ which maps $n$ to the number of distinct factors of length $n$ occurring in the infinite word $\mathbf{x}$, counted up to $\ell$-abelian equivalence.  That is, we count $\ell$-abelian equivalence classes partitioning the set of factors $\fac_{\mathbf{x}}(n)$ of length $n$ occurring in $\mathbf{x}$. In particular, for any infinite word $\mathbf{x}$, we have for all $n\ge 0$ $$\mathcal{P}^{(1)}_\mathbf{x}(n)\le \cdots \le \mathcal{P}^{(\ell)}_\mathbf{x}(n)\le \mathcal{P}^{(\ell+1)}_\mathbf{x} (n)\le \cdots\le \mathcal{P}^{(\infty)}_\mathbf{x}(n).$$

In this paper, we show that both the period-doubling word and the Thue--Morse word have $2$-abelian complexity sequences which are $2$-regular.
The computations and arguments leading to these results permit us to exhibit some similarities between the two cases and a quite general scheme that we hope can be used again to prove additional regularity results. Indeed, one conjectures that {\em any $k$-automatic sequence has an $\ell$-abelian complexity function that is $k$-regular}.

We mention some other papers containing related work.
In \cite{Karhumaki--Saarela--Zamboni arxiv}, the authors studied the asymptotic behavior of $\mathcal{P}^{(\ell)}_\mathbf{t}(n)$ and also derived some recurrence relations\footnote{It seems that there is some subtle error in the relation for $\mathcal{P}^{(1)}_\mathbf{p}(4n+2)$ proposed in \cite[Lemma~6]{Karhumaki--Saarela--Zamboni arxiv}. Correct relations are given by \cite[Proposition~2]{Blanchet} and could also be obtained by Theorem~\ref{thm:reflection_recurrence} and Proposition~\ref{prop:recdelta_TM}.} showing that the abelian complexity $\mathcal{P}^{(1)}_\mathbf{p}(n)_{n\ge 0}$ of the period-doubling word $\mathbf{p}$ is $2$-regular.
In \cite{Blanchet}, the abelian complexity of the fixed point $\mathbf{v}$ of the non-uniform morphism $0\mapsto 012, 1\mapsto 02, 2\mapsto 1$ is studied and the authors obtain results similar to those discussed in this paper. Even though the authors of \cite{Blanchet} are not directly interested in the $k$-regularity of $\mathcal{P}^{(1)}_\mathbf{v}(n)_{n\ge 0}$, they derive recurrence relations. From these relations, following the approach described in this paper, one can possibly prove some regularity result. In particular, the result of replacing  in $\mathbf{v}$ all $2$'s by $0$'s leads back to the period-doubling word. Hence, Blanchet-Sadri et~al.\ also proved some other relations about the abelian complexity of $\mathbf{p}$.

Given the first few terms of a sequence, one can easily conjecture the potential $k$-regularity of this sequence by exhibiting relations that should be satisfied; see \cite[Section~6]{Allouche--Shallit 2003} for such a ``predictive'' algorithm that recognizes regularity. Of course, in such an algorithm, a finite examination does not lead to a proof of the $k$-regularity of a sequence. The first few terms of the $2$-abelian complexity $\mathcal{P}^{(2)}_\mathbf{t}(n)_{n\ge 0}$ of the Thue--Morse word are
$$1, 2, 4, 6, 8, 6, 8, 10, 8, 6, 8, 8, 10, 10, 10, 8, 8, 6, 8, 10, 10, 8, 10, 12, 12, 10, 12, 12, \ldots.$$
The second and last authors of this paper conjectured the $2$-regularity of the sequence $\mathcal{P}^{(2)}_\mathbf{t}(n)_{n\ge 0}$ (and proved some recurrence relations for this sequence)~\cite{Rigo Vandomme}. Recently, after hearing a talk given by the last author during the \textit{Representing Streams~II} meeting in January 2014, Greinecker proved the recurrence relations needed to prove the $2$-regularity of this sequence~\cite{Greinecker}. Hopefully, the two approaches are complementary: in this paper, we prove $2$-regularity without exhibiting the explicit recurrence relations.
\smallskip

Let us now describe the content and organization of this paper.
\smallskip

In Section~\ref{sec:2} we prove Theorem~\ref{thm:reflection_recurrence}, which establishes the $2$-regularity of a large family of sequences satisfying a recurrence relation with a parameter $c$ and $2^{\ell_0}$ initial conditions. The form of the recurrence implies that sequences in this family exhibit a reflection symmetry in the values taken over each interval $[2^\ell,2^{\ell+1})$ for $\ell \geq \ell_0$. For the special case of the Thue--Morse word, a similar property is shown in \cite{Greinecker}. Computer experiments suggest that many $2$-abelian complexity functions satisfy such a reflection property.
\begin{theorem}\label{thm:reflection_recurrence}
Let $\ell_0 \geq 0$ and $c \in \Z$.
Suppose $s(n)_{n \geq 0}$ is a sequence such that, for all $\ell \geq \ell_0$ and for all $r$ such that $0 \leq r \leq 2^\ell - 1$, we have
\begin{equation}\label{eq:eric}
	s(2^\ell + r) =
	\begin{cases}
		s(r) + c	& \text{if $r \leq 2^{\ell-1}$} \\
		s(2^{\ell+1}-r)	& \text{if $r > 2^{\ell-1}$}.
	\end{cases}
\end{equation}
Then $s(n)_{n \geq 0}$ is $2$-regular.
\end{theorem}

The recurrence satisfied by $s(n)$ in Theorem~\ref{thm:reflection_recurrence} reads words from left to right, i.e., starting with the most significant digit.
Our proof of this theorem will express sequences in the $2$-kernel of $s(n)_{n \geq 0}$ as in Definition~\ref{def:regular}, starting with the least significant digit.

From Equation~\eqref{eq:eric} one can get some information about the asymptotic behavior of the sequence $s(n)_{n \geq 0}$.
We have $s(n) = O(\log n)$, and moreover
\[
	s\left(\tfrac{4^{\ell + 1} - 1}{3}\right) = s(4^\ell + \cdots + 4^1 + 4^0) = \left(\ell - \left\lfloor\tfrac{\ell_0 - 1}{2}\right\rfloor\right) c + s\left(\tfrac{4^{\lfloor(\ell_0 + 1)/2\rfloor} - 1}{3}\right)
\]
for $\ell \geq \lfloor\frac{\ell_0 - 1}{2}\rfloor$.
At the same time, there are many subsequences of $s(n)_{n \geq 0}$ which are constant; for example, $s(2^\ell) = c$ for $\ell \geq \ell_0$.

\begin{example}
    As an illustration of the reflection property described in Theorem~\ref{thm:reflection_recurrence}, we consider the abelian complexity of the $2$-block coding of the period-doubling word $\mathbf{p}$. (The recurrence satisfied by this sequence is given in Theorem~\ref{thm:recab_PD}.) Some values of this sequence are depicted in Figures~\ref{fig:reflect1} and \ref{fig:reflect2}.
    \begin{figure}[h!tbp]
        \centering
        \scalebox{.6}{\includegraphics{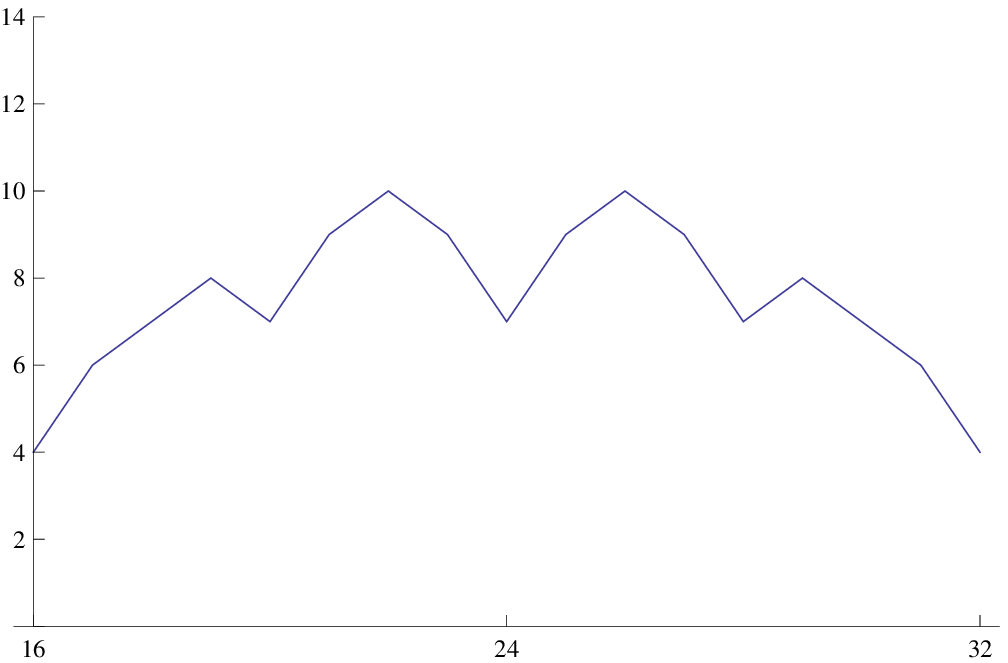}\quad \includegraphics{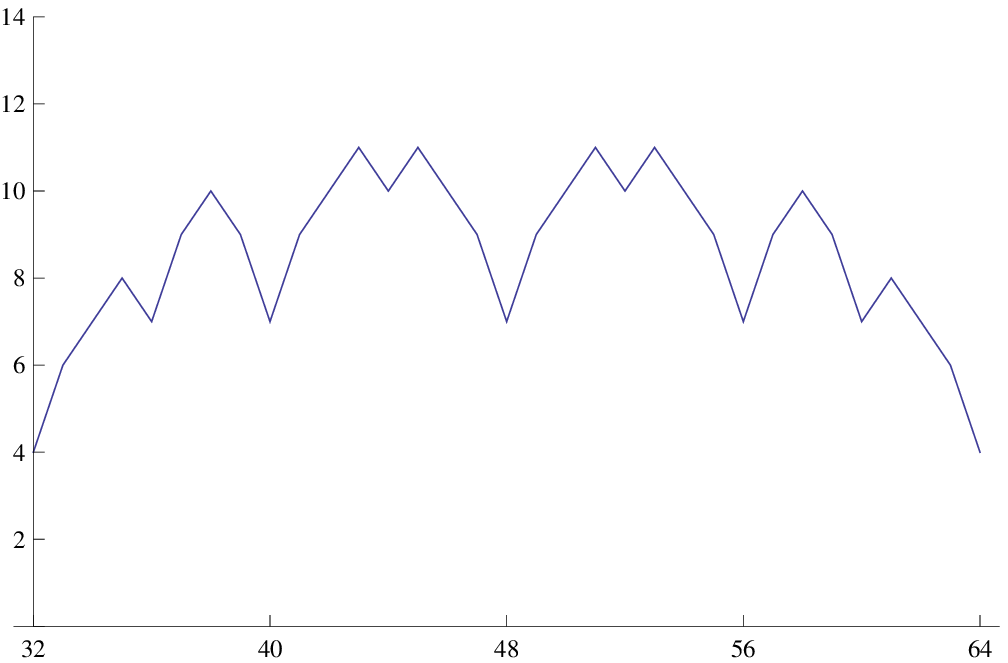}}
        \caption{The abelian complexity of $\blo(\mathbf{p},2)$ on the intervals $[16,32]$ and $[32,64]$.}
        \label{fig:reflect1}
    \end{figure} 
\begin{figure}[h!tbp]
        \centering
        \scalebox{.8}{\includegraphics{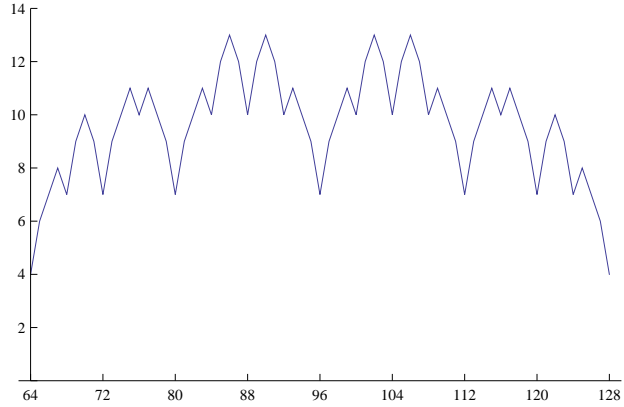}}
        \caption{The abelian complexity of $\blo(\mathbf{p},2)$ on the interval $[64,128]$.}
        \label{fig:reflect2}
    \end{figure}
\end{example}

In Section~\ref{sec:3}, we collect some general results and definitions about words and $k$-regular sequences (in particular stability properties of the set of $k$-regular sequences under sum and product) that are needed in the other parts of this paper. 

In Section~\ref{sec:block_PD}, we study the abelian complexity of the $2$-block coding $\mathbf{x}=\blo(\mathbf{p},2)$ of the period-doubling word $\mathbf{p}$. In particular, we consider the difference $\Delta_{0}(n)$ between the maximal and minimal numbers of $0$'s occurring in factors of length $n$ in $\blo(\mathbf{p},2)$. We prove that the sequences $\Delta_{0}(n)_{n\ge 0}$ and $\mathcal{P}^{(1)}_\mathbf{x}(n)_{n\ge 0}$ are $2$-regular.
In Section~\ref{sec:2ab_comp_PD}, we study the $2$-abelian complexity of $\mathbf{p}$. We show that the $2$-regularity of $\Delta_{0}(n)_{n\ge 0}$ and $\mathcal{P}^{(1)}_\mathbf{x}(n)_{n\ge 0}$ implies the $2$-regularity of $\mathcal{P}^{(2)}_\mathbf{p}(n)$.

Sections~\ref{sec:block_TM} and \ref{sec:2ab_comp_TM} share some similarities with Sections~\ref{sec:block_PD} and \ref{sec:2ab_comp_PD}. The reader will see that the strategy used to prove the  $2$-regularity of $\mathcal{P}^{(2)}_\mathbf{p}(n)$ can also be applied to the Thue--Morse word. Nevertheless, some differences do not permit us to treat the two cases within a completely unified framework.

In Section~\ref{sec:block_TM}, we study the abelian complexity of the $2$-block coding $\mathbf{y}=\blo(\mathbf{t},2)$ of the Thue--Morse word $\mathbf{t}$. We define $\Delta_{12}(n)$ to be the difference between the maximal total and minimal total numbers of $1$'s and $2$'s occurring in factors of length $n$ in $\blo(\mathbf{t},2)$. It turns out that $\Delta_{12}(n) + 1 = \mathcal{P}^{(1)}_\mathbf{p}(n)$ and our results can thus be related to \cite{Blanchet} and \cite{Karhumaki--Saarela--Zamboni arxiv}. We prove that $\Delta_{12}(n)_{n\ge 0}$ and $\mathcal{P}^{(1)}_\mathbf{y}(n)_{n\ge 0}$ are $2$-regular.
In Section~\ref{sec:2ab_comp_TM}, we show that the $2$-regularity of $\mathcal{P}^{(2)}_\mathbf{t}(n)$ follows from the $2$-regularity of $\Delta_{12}(n)_{n\ge 0}$ and $\mathcal{P}^{(1)}_\mathbf{y}(n)_{n\ge 0}$.

Finally, in Section~\ref{Conclusions} we suggest a direction for future work.

\tableofcontents


\section{Sequences satisfying a reflection symmetry}\label{sec:2}
The aim of this section is to prove Theorem~\ref{thm:reflection_recurrence} stated in the introduction. 
Before proving it in generality, we first examine the sequence satisfying the recurrence for $\ell_0 = 0$ and $c = 1$.
It will turn out that the general solution can be expressed naturally in terms of this sequence.

Let $A(0) = 0$.
For each $\ell \geq 0$ and $0 \leq r \leq 2^\ell - 1$, let
\begin{equation}\label{recurrence for A}
	A(2^\ell + r) =
	\begin{cases}
		A(r) + 1		& \text{if $r \leq 2^{\ell-1}$} \\
		A(2^{\ell+1} - r)	& \text{if $r > 2^{\ell-1}$}.
	\end{cases}
\end{equation}
The sequence $A(n)_{n \geq 0}$ is
\[
	0, 1, 1, 2, 1, 2, 2, 2, 1, 2, 2, 3, 2, 3, 2, 2, \dots
\]
and appears as \seq{A007302}.
Allouche and Shallit~\cite[Example~12]{Allouche--Shallit 2003} identified this sequence as an example of a regular sequence.
We include a proof here.

\begin{proposition}\label{A007302 is 2-regular}
For all $n \geq 0$ we have
\begin{align*}
	A(2 n) &= A(n) \\
	A(8 n + 1) &= A(4 n + 1) \\
	A(8 n + 3) &= A(2 n + 1) + 1 \\
	A(8 n + 5) &= A(2 n + 1) + 1 \\
	A(8 n + 7) &= A(4 n + 3).
\end{align*}
In particular, $A(n)_{n \geq 0}$ is $2$-regular.
\end{proposition}

\begin{proof}
This proof is typical of many of the proofs throughout the paper.
We work by induction on $n$.
The case $n = 0$ can be checked easily using the first few values of the sequence $A(n)_{n \geq 0}$.
Therefore, let $n \geq 1$ and assume that the recurrence holds for all values less than $n$.
Write $n = 2^\ell + r$ with $\ell \geq 0$ and $0 \leq r \leq 2^\ell - 1$.

First let us address the equation $A(2 n) = A(n)$.
If $0 \leq r \leq 2^{\ell - 1}$, then
\begin{align*}
	A(2 n)
	&= A(2^{\ell + 1} + 2 r) \\
	&= A(2 r) + 1			& \text{(by Equation~\eqref{recurrence for A})} \\
	&= A(r) + 1			& \text{(by induction hypothesis)} \\
	&= A(2^\ell + r)		& \text{(by Equation~\eqref{recurrence for A})} \\
	&= A(n).
\end{align*}
On the other hand, if $2^{\ell - 1} < r < 2^\ell$, then
\begin{align*}
	A(2 n)
	&= A(2^{\ell + 1} + 2 r) \\
	&= A(2^{\ell + 2} - 2 r)	& \text{(by Equation~\eqref{recurrence for A})} \\
	&= A(2^{\ell + 1} - r)		& \text{(by induction hypothesis)} \\
	&= A(2^\ell + r)		& \text{(by Equation~\eqref{recurrence for A})} \\
	&= A(n).
\end{align*}

Next we consider $A(8 n + 1) = A(4 n + 1)$.
If $0 \leq r \leq 2^{\ell - 1} - 1$, then
\begin{align*}
	A(8 n + 1)
	&= A(2^{\ell + 3} + 8 r + 1) \\
	&= A(8 r + 1) + 1	& \text{(by Equation~\eqref{recurrence for A})} \\
	&= A(4 r + 1) + 1		& \text{(by induction hypothesis)} \\
	&= A(2^{\ell + 2} + 4 r + 1)	& \text{(by Equation~\eqref{recurrence for A})} \\
	&= A(4 n + 1).
\end{align*}
If $2^{\ell - 1} \leq r < 2^\ell$, then
\begin{align*}
	A(8 n + 1)
	&= A(2^{\ell + 3} + 8 r + 1) \\
	&= A(2^{\ell + 4} - 8 r - 1)		& \text{(by Equation~\eqref{recurrence for A})} \\
	&= A(2^{\ell + 4} - 8 r - 8 + 7) \\
	&= A(2^{\ell + 3} - 4 r - 4 + 3)	& \text{(by induction hypothesis)} \\
	&= A(2^{\ell + 3} - (4 r + 1)) \\
	&= A(2^{\ell + 2} + 4 r + 1)		& \text{(by Equation~\eqref{recurrence for A})} \\
	&= A(4 n + 1).
\end{align*}
The equations for $A(8 n + 3)$, $A(8 n + 5)$ and $A(8 n + 7)$ are handled similarly.
\end{proof}

Now we prove Theorem~\ref{thm:reflection_recurrence}.
We show that for general $\ell_0 \geq 0$, a sequence $s(n)_{n \geq 0}$ satisfying the recurrence can be written in terms of $A(n)_{n \geq 0}$.

\begin{proof}[Proof of Theorem~\ref{thm:reflection_recurrence}]
There are $2^{\ell_0}$ initial conditions for the recurrence, namely $s(0)$, \ldots, $s(2^{\ell_0} - 1)$.
We claim that most of the $2^{\ell_0 + 2}$ subsequences of the form $s(2^{\ell_0 + 2} n + i)_{n \geq 0}$ depend on only one of the initial conditions $s(j)$; each of these subsequences is essentially $A(n)_{n \geq 0}$, $A(4 n + 1)_{n \geq 0}$, $A(2 n + 1)_{n \geq 0}$, or $A(4 n + 3)_{n \geq 0}$.
Furthermore, each of the remaining subsequences is equal to $s(2^{\ell_0} n + j) + c$ for some $j$.
More precisely, for $0 \leq i \leq 2^{\ell_0 + 2} - 1$ and $n \geq 0$ we have the identity
\begin{align*}
	s(&2^{\ell_0 + 2}n + i)= \\
	&\begin{cases}
		c \, A(n) + s(0)				& \text{if $i = 0$} \\
		c \, A(4 n + 1) - c + s(i)			& \text{if $1 \leq i \leq 2^{\ell_0} - 1$} \\
		c \, A(4 n + 1) + s(0)				& \text{if $i = 2^{\ell_0}$} \\
		s(2^{\ell_0} n + i - 2^{\ell_0}) + c		& \text{if $2^{\ell_0} + 1 \leq i \leq 2^{\ell_0} + 2^{\ell_0 - 1} - 1$} \\
		c \, A(2 n + 1) + s(|i - 2^{\ell_0 + 1}|)	& \text{if $2^{\ell_0} + 2^{\ell_0 - 1} \leq i \leq 2^{\ell_0 + 1} + 2^{\ell_0 - 1}$} \\
		s(2^{\ell_0} n + i - 2^{\ell_0 + 1}) + c	& \text{if $2^{\ell_0 + 1} + 2^{\ell_0 - 1} + 1 \leq i \leq 2^{\ell_0 + 1} + 2^{\ell_0} - 1$} \\
		c \, A(4 n + 3) + s(0)				& \text{if $i = 2^{\ell_0 + 1} + 2^{\ell_0}$} \\
		c \, A(4 n + 3) - c + s(2^{\ell_0 + 2} - i)	& \text{if $2^{\ell_0 + 1} + 2^{\ell_0} + 1 \leq i \leq 2^{\ell_0 + 2} - 1$}.
	\end{cases}
\end{align*}
(Note the symmetry among the eight cases, which reflects the symmetry $s(2^\ell + r) = s(2^{\ell+1}-r)$ of the recurrence for $r > 2^{\ell - 1}$.)
It will follow from this identity that the $\Z$-module generated by the $2$-kernel of $s(n)_{n \geq 0}$ is generated by the sequences $s(2^\ell n + j)_{n \geq 0}$ for $0 \leq \ell \leq \ell_0 + 1$ and $0 \leq j \leq 2^\ell - 1$, $A(n)_{n \geq 0}$, $A(4 n + 1)_{n \geq 0}$, $A(2 n + 1)_{n \geq 0}$, $A(4 n + 3)_{n \geq 0}$, and the constant $1$ sequence.
In particular, this module is finitely generated.

We prove the identity by induction on $n$.
Recall that for all $\ell \geq \ell_0$ and for all $r$ such that $0 \leq r \leq 2^\ell - 1$, we have Equation~\eqref{eq:eric}, i.e., 
\[
	s(2^\ell + r) =
	\begin{cases}
		s(r) + c	& \text{if $r \leq 2^{\ell-1}$} \\
		s(2^{\ell+1}-r)	& \text{if $r > 2^{\ell-1}$}.
	\end{cases}
\]
For $n = 0$, one uses $A(1) = 1$ and $A(3) = 2$ to verify that all eight cases of the identity hold.
Inductively, let $n \geq 1$, and assume the identity is true for all $n' < n$.
Write $n = 2^\ell + r$ with $\ell \geq 0$ and $0 \leq r \leq 2^\ell - 1$.

First we consider the case $0 \leq r \leq 2^{\ell - 1} - 1$.
For all $i\in\{0,\ldots 2^{\ell_0 + 2} - 1\}$, we have $2^{\ell_0 + 2} r + i \leq 2^{(\ell_0 + 2 + \ell) - 1} - 1$, so
\begin{align*}
	s(2^{\ell_0 + 2} n + i)
	&= s(2^{\ell_0 + 2 + \ell} + (2^{\ell_0 + 2} r + i)) \\
	&= s(2^{\ell_0 + 2} r + i) + c		& \text{(by Equation~\eqref{eq:eric})}.
\end{align*}
If $1 \leq i \leq 2^{\ell_0} - 1$, then the induction hypothesis now gives
\begin{align*}
	s(2^{\ell_0 + 2} n + i)
	&= s(2^{\ell_0 + 2} r + i) + c \\
	&= c \, A(4 r + 1) + s(i) \\
	&= c \left(A(2^{\ell + 2} + 4 r + 1) - 1\right) + s(i) \\
	&= c \, A(4 n + 1) - c + s(i),
\end{align*}
where we have used $A(2^{\ell + 2} + 4 r + 1) = A(4 r + 1) + 1$ from the recurrence for $A(n)$, since $4 r + 1 \leq 2^{(\ell + 2) - 1}$.
The other seven intervals for $i$ are verified similarly; in each case one applies the induction hypothesis to $s(2^{\ell_0 + 2} r + i) + c$ and then uses the recurrence for either $A(n)$ or $s(n)$ to raise an argument in $r$ to an argument in $n$.

It remains to consider $2^{\ell - 1} \leq r \leq 2^\ell - 1$.
First we address the case $i = 0$.
If $r = 2^{\ell - 1}$ then
\begin{align*}
	s(2^{\ell_0 + 2} n + i)
	&= s(2^{\ell_0 + 2 + \ell} + 2^{\ell_0 + 2 + \ell - 1}) \\
	&= s(2^{\ell_0 + 2 + \ell - 1}) + c				& \text{(by Equation~\eqref{eq:eric})} \\
	&= c \, A(2^{\ell - 1}) + s(0) + c				& \text{(by inductive hypothesis)} \\
	&= c \left(A(2^\ell + 2^{\ell - 1}) - 1\right) + s(0) + c	& \text{(by Equation~\eqref{recurrence for A})} \\
	&= c \, A(n) + s(0)
\end{align*}
as desired.
Alternatively, if $2^{\ell - 1} < r \leq 2^\ell - 1$ then $2^{\ell_0 + 2} r > 2^{(\ell_0 + 2 + \ell) - 1}$, so
\begin{align*}
	s(2^{\ell_0 + 2} n + i)
	&= s(2^{\ell_0 + 2 + \ell} + 2^{\ell_0 + 2} r) \\
	&= s(2^{\ell_0 + 2 + \ell + 1} - 2^{\ell_0 + 2} r)	& \text{(by Equation~\eqref{eq:eric})} \\
	&= s(2^{\ell_0 + 2} (2^{\ell + 1} - r) + 0) \\
	&= c \, A(2^{\ell + 1} - r) + s(0)			& \text{(by inductive hypothesis)} \\
	&= c \, A(2^\ell + r) + s(0)				& \text{(by Equation~\eqref{recurrence for A})} \\
	&= c \, A(n) + s(0).
\end{align*}

Therefore it remains to consider $2^{\ell - 1} \leq r \leq 2^\ell - 1$ for $1 \leq i \leq 2^{\ell_0 + 2} - 1$.
In this range we have $2^{\ell_0 + 2} r + i > 2^{(\ell_0 + 2 + \ell) - 1}$, so
\begin{align*}
	s(2^{\ell_0 + 2} n + i)
	&= s(2^{\ell_0 + 2 + \ell} + (2^{\ell_0 + 2} r + i)) \\
	&= s(2^{\ell_0 + 2 + \ell + 1} - 2^{\ell_0 + 2} r - i)	& \text{(by Equation~\eqref{eq:eric})} \\
	&= s(2^{\ell_0 + 2} n' + i'),
\end{align*}
where $n' = 2^{\ell + 1} - r - 1$ and $i' = 2^{\ell_0 + 2} - i$.
We prove the identity for the seven intervals for $i$ using the same steps we have already used several times; we have just applied the recurrence for $s(n)$, so next we use the induction hypothesis, followed by the recurrence for $A(n)$ or $s(n)$, depending on which term appears.
For the first interval, if $1 \leq i \leq 2^{\ell_0} - 1$, then $2^{\ell_0 + 1} + 2^{\ell_0} + 1 \leq i' \leq 2^{\ell_0 + 2} - 1$, so
\begin{align*}
	s(2^{\ell_0 + 2} n + i)
	&= s(2^{\ell_0 + 2} n' + i') \\
	&= c \, A(4 n' + 3) - c + s(2^{\ell_0 + 2} - i')	& \text{(by inductive hypothesis)} \\
	&= c \, A(2^{\ell + 3} - (4 r + 1)) - c + s(i) \\
	&= c \, A(2^{\ell + 2} + 4 r + 1) - c + s(i)		& \text{(by Equation~\eqref{recurrence for A})} \\
	&= c \, A(4 n + 1) - c + s(i).
\end{align*}
The proofs for the remaining six intervals are routine at this point, so we omit the steps here.
\end{proof}

\begin{example}\label{exa:l0=2}
In Section~\ref{sec:block_PD}, we will use Theorem~\ref{thm:reflection_recurrence} with $\ell_0 = 2$ to conclude that $\Delta_0(n)_{n \geq 0}$ and $\mathcal{P}^{(1)}_{\mathbf{x}}(n)_{n \geq 0}$ are $2$-regular for the period-doubling word.
For $\ell_0 = 2$ the value of $s(16 n + i)$ is
\[
	s(16 n + i) =
	\begin{cases}
		c \, A(n) + s(0)		& \text{if $i = 0$} \\
		c \, A(4 n + 1) - c + s(i)	& \text{if $1 \leq i \leq 3$} \\
		c \, A(4 n + 1) + s(0)		& \text{if $i = 4$} \\
		s(4 n + 1) + c			& \text{if $i = 5$} \\
		c \, A(2 n + 1) + s(|i - 8|)	& \text{if $6 \leq i \leq 10$} \\
		s(4 n + 3) + c			& \text{if $i = 11$} \\
		c \, A(4 n + 3) + s(0)		& \text{if $i = 12$} \\
		c \, A(4 n + 3) - c + s(16 - i)	& \text{if $13 \leq i \leq 15$}.
	\end{cases}
\]
In Section~\ref{sec:block_TM}, we will use Theorem~\ref{thm:reflection_recurrence} with $\ell_0 = 1$ to conclude that $\Delta_{12}(n)_{n \geq 0}$ is $2$-regular for the Thue--Morse word.
\end{example}


\section{About regular sequences and words}\label{sec:3}

We will often make use of the following composition theorem for a function $F$ defined piecewise on several $k$-automatic sets.

\begin{lemma}\label{lem:compo}
Let $k \geq 2$.
Let $P_1,\ldots,P_\ell:\mathbb{N}\to\{0,1\}$ be unary predicates that are $k$-automatic. Let $f_1,\ldots,f_\ell$ be $k$-regular functions. The function $F:\mathbb{N}\to\mathbb{N}$ defined by
$$F(n)=\sum_{i=1}^\ell f_i(n)\, P_i(n)$$
is $k$-regular.
\end{lemma}

\begin{proof}
    It is a direct consequence of \cite[Theorem~2.5]{Allouche--Shallit 1992}: if $s(n)_{n\ge 0}$ and $t(n)_{n\ge 0}$ are $k$-regular, then $(s(n)+t(n))_{n\ge 0}$ and $(s(n)t(n))_{n\ge 0}$ are both $k$-regular sequences. Recall that $k$-automatic sequences are special cases of $k$-regular sequences.
\end{proof}

Note that if, for each $n$, there is exactly one $i$ such that $P_i(n) = 1$, then we can write
\[
	F(n) =
	\begin{cases}
		f_1(n)		& \text{if $P_1(n) = 1$} \\
		f_2(n)		& \text{if $P_2(n) = 1$} \\
		\quad \vdots	& \qquad \vdots \\
		f_\ell(n)	& \text{if $P_\ell(n) = 1$}.
	\end{cases}
\]
This is the setting in which we will apply Lemma~\ref{lem:compo}.

We will also make use of the following classical results.

\begin{lemma}\cite[Theorem~2.3]{Allouche--Shallit 1992}
     Let $k \geq 2$ be an integer. A sequence taking finitely many values is $k$-regular if and only if it is $k$-automatic.
\end{lemma}

\begin{lemma}\cite[Corollary~2.4]{Allouche--Shallit 1992}\label{lem:regmod}
    Let $k,m\ge 2$ be integers. If a sequence $s(n)_{n\ge 0}$ is $k$-regular, then $(s(n)\bmod{m})_{n\ge 0}$ is $k$-automatic.
\end{lemma}

\begin{lemma}\label{lem:shift}
    Let $k \geq 2$ be an integer. Let $s(n)_{n\ge 0}$ be a sequence.  The sequence $s(n)_{n\ge 0}$ is $k$-regular if and only if $s(n + 1)_{n\ge 0}$ is $k$-regular.
\end{lemma}

\begin{proof}
It is a direct consequence of two results stated in \cite{Allouche--Shallit 1992}, namely  Theorem~2.6 and its following remark.
\end{proof}

Let us now give some definitions about combinatorics on words.

\begin{definition}
If a word $w$ starts with the letter $a$, then $a^{-1}w$ denotes the word obtained from $w$ by deleting its first letter. Similarly, if a word $w$ ends with the letter $a$, then $wa^{-1}$ denotes the word obtained from $w$ by deleting its last letter. As usual, we let $|w|$ denote the length of the finite word $w$.
If $a$ is a letter, we let $|w|_a$ denote the number of occurrences of $a$ in $w$.
If $w = w_0 \cdots w_{\ell-1}$, then we let $w^\mathrm{R} = w_{\ell-1} \cdots w_0$ denote the reversal of $w$.
Our convention is that we index letters in an infinite word beginning with $0$.
\end{definition}

Since we are interested in $\ell$-abelian complexity, it is natural to consider the following operation that permits us to compare factors of length $\ell$ occurring in an infinite word. Indeed, if two finite words are $\ell$-abelian equivalent, then their $\ell$-block codings are abelian equivalent (but the converse does not hold).

\begin{definition}\label{def:block}
Let $\ell\ge 1$. The {\em $\ell$-block coding} of the word $\mathbf{w}=w_0w_1w_2\cdots$ over the alphabet $A$ is the word 
$$\blo(\mathbf{w},\ell)=(w_0\cdots w_{\ell-1})\, (w_1\cdots w_\ell)\, (w_2\cdots w_{\ell+1})\cdots (w_j\cdots w_{j+\ell-1})\cdots$$
over the alphabet $A^\ell$. If $A=\{0,\ldots, r-1\}$, then it is convenient to identify $A^\ell$ with the set $\{0,\ldots, r^\ell-1\}$ and each word $w_0\cdots w_{\ell-1}$ of length $\ell$ is thus replaced with the integer obtained by reading the word in base $r$, i.e., 
$$\sum_{i=0}^{\ell-1} w_i\, r^{\ell-1-i}.$$
One can also define accordingly the $\ell$-block coding of a finite word $u$ of length at least $\ell$. The resulting word $\blo(u,\ell)$ has length $|u|-\ell+1$.
\end{definition}

\begin{example}
    The $2$-block codings of $011010011$ and $001101101$ are respectively $13212013$ and $01321321$, which are abelian equivalent.
\end{example}

\begin{lemma}\cite[Lemma~2.3]{Karhumaki--Saarela--Zamboni}\label{lem:abel}
    Let $\ell\ge 1$. Two finite words $u$ and $v$ of length at least $\ell-1$ are $\ell$-abelian equivalent if and only if they share the same prefix (resp.\ suffix) of length $\ell-1$ and the words $\blo(u,\ell)$ and $\blo(v,\ell)$ are abelian equivalent.
\end{lemma}

It is well known that the $\ell$-block coding of a $k$-automatic sequence is again a $k$-automatic sequence~\cite{Cobham 1972}. (Note that the operation of $\ell$-block compression that one also encounters in the literature is not the same as the $\ell$-block coding given in Definition~\ref{def:block}.)

\begin{example}\label{exa:cp2}
For the period-doubling word $\mathbf{p}$, the $2$-block coding is given by 
\[
    	\blo(\mathbf{p},2)=\phi^\omega(1) = 12001212120012001200121212001212 \cdots
\]
where $\phi$ is the morphism over $\{0,1,2\}^*$ defined by $\phi : 0 \mapsto 12, 1\mapsto 12, 2\mapsto 00$. 
\end{example}

\begin{example}\label{exa:ct2}
    For the Thue--Morse word $\mathbf{t}$,  the $2$-block coding is given by 
\[\blo(\mathbf{t},2)=\nu^\omega(1)=132120132012132120121320\cdots
\]
where $\nu$ is the morphism over $\{0,1,2,3\}^*$ defined by $\nu:0\mapsto 12,1\mapsto 13$, $2\mapsto 20,3 \mapsto 21$. 
\end{example}


\section{Abelian complexity of $\blo(\mathbf{p},2)$}\label{sec:block_PD}
We let $\mathbf{x}$ denote $\blo(\mathbf{p},2)=12001212120012001200121212001212\cdots$, the $2$-block coding of $\mathbf{p}$, introduced in Example~\ref{exa:cp2}. We consider in this section the abelian complexity of $\mathbf{x}$ and then, in Section~\ref{sec:2ab_comp_PD}, we compare $\mathcal{P}^{(1)}_{\mathbf{x}}(n)$ with $\mathcal{P}^{(2)}_{\mathbf{p}}(n)$.

\begin{definition}\label{def:M0}
We will make use of functions related to the number of $0$'s in the factors of $\mathbf x$ of a given length.
Let $n\in \mathbb N$. We let $\M_0(n)$ (resp.\ $\m_0(n)$) denote the maximum (resp.\ minimum) number of $0$'s in a factor of $\mathbf x$ of length $n$.
Let $\Delta_0(n)=\M_0(n)-\m_0(n)$ be the difference between these two values. 
\end{definition}

Each of the $\Delta_0(n) + 1$ integers in the interval $[\m_0(n), \M_0(n)]$ is attained as the number of $0$'s in some factor of $\mathbf x$ of length $n$, since when we slide a window of length $n$ along $\mathbf x$ from a factor with $\m_0(n)$ zeros to a factor with $\M_0(n)$ zeros, the number of $0$'s changes by at most $1$ per step.

\begin{lemma}\label{lem:minmaxeven}
If $n$ is even, then $\M_0(n)$, $\m_0(n)$ and $\Delta_0(n)$ are even. 
\end{lemma}

\begin{proof}
Suppose a factor $w=w_1\cdots w_{2n}$ of $\mathbf{x}$ of even length $2n$ has an odd number $n_0$ of zeros.
Since $\phi(0)=\phi(1)=12$ and $\phi(2)=00$, the factor $w$ starts or ends with $0$. Without loss of generality, assume it starts with $w_1=0$. Then its last letter must be $w_{2n}=1$. The words $0w_1\cdots w_{2n-1}$ and $w_2\cdots w_{2n}2$ are two factors of length $2n$ with respectively $n_0+1$ and $n_0-1$ zeros. Hence, these two factors have even numbers of zeros which are respectively greater than and less than $n_0$. The conclusion follows.
\end{proof}

We give two related proofs of the $2$-regularity of the sequence $\mathcal{P}^{(1)}_{\mathbf{x}}(n)_{n\ge 0}$. The first uses the following proposition, which we prove in Section~\ref{sec:51}, together with the fact that $\Delta_0(n)_{n\ge 0}$ is $2$-regular and the two sequences $(\Delta_0(n)\bmod{2})_{n\ge 0}$ and  $(\m_0(n)\bmod{2})_{n\ge 0}$ are $2$-automatic (see Section~\ref{sec:52}, Corollary~\ref{cor:2regDm_PD}). Then the $2$-regularity of the sequence $\mathcal{P}^{(1)}_{\mathbf{x}}(n)_{n\ge 0}$ will follow from Lemma~\ref{lem:compo}.
\begin{proposition}\label{prop:deltatoab_PD}
  For $n\in \mathbb N$,
 $$\mathcal{P}^{(1)}_{\mathbf{x}}(n)=
  \begin{cases}
    \frac{3}{2}\Delta_0(n)+\frac{3}{2} & \text{if } \Delta_0(n) \text{ is odd} \\
    \frac{3}{2}\Delta_0(n)+1 & \text{if } \Delta_0(n) \text{ and } n-\m_0(n) \text{ are even}\\
    \frac{3}{2}\Delta_0(n)+2 & \text{if } \Delta_0(n) \text{ and } n-\m_0(n)+1 \text{ are even}.\\
   \end{cases}
$$
\end{proposition}

In the second proof, we prove in Section~\ref{sec:53} the following theorem, which allows us to apply our general result expressed by Theorem~\ref{thm:reflection_recurrence}.
\begin{theorem}\label{thm:recab_PD}
Let $\ell\geq 2$ and $r$ such that $0\leq r <2^{\ell}-1$. We have
  $$\mathcal{P}^{(1)}_{\mathbf{x}}(2^\ell+r)=
  \begin{cases}
    \mathcal{P}^{(1)}_{\mathbf{x}}(r)+3 & \text{if } r\leq 2^{\ell-1} \\
    \mathcal{P}^{(1)}_{\mathbf{x}}(2^{\ell+1}-r) & \text{if } r>2^{\ell-1}. \\
   \end{cases}
$$
In particular, the sequence $\mathcal{P}^{(1)}_{\mathbf{x}}(n)_{n\ge 0}$ is $2$-regular.
\end{theorem}

From Theorem~\ref{thm:recab_PD} we see that $\mathcal{P}^{(1)}_{\mathbf{x}}(2^\ell) = \mathcal{P}^{(1)}_{\mathbf{x}}(0) + 3 = 4$ for all $\ell \geq 2$.
Additionally, one can check that $\mathcal{P}^{(1)}_{\mathbf{x}}(2^1) = 4$.


\subsection{Proof of Proposition~\ref{prop:deltatoab_PD}}\label{sec:51}

First we mention some properties of factors of the word $\mathbf{x}$.

\begin{lemma}\label{lem:2facPD}
The set of factors of $\mathbf{x}$ of length $2$ is $\fac_{\mathbf{x}}(2)=\{00,01,12,20,21\}$.
\end{lemma}

\begin{proof}
It is easy to check that these five words are factors. To prove that they are the only ones, it is enough to check that for any element $u$ in $\{00,01,12,20,21\}$ the three factors of length $2$ of $\phi(u)$ are in $\{00,01,12,20,21\}$.
\end{proof}

\begin{lemma}\label{lem:balancedPD}
If $w$ is a factor of $\mathbf{x}$ then $\big||w|_1 - |w|_2\big| \leq 1$. In particular, the letters $1$ and $2$ alternate in the sequence obtained from $\mathbf{x}$ after erasing the $0$'s.
\end{lemma}

\begin{proof}
Let $w$ be a factor of $\mathbf{x}$. There are two cases to consider.

If $w$ can be de-substituted (that is, $w = \phi(v)$ for some $v$), then $|w|_1 = |w|_2$ since $|\phi(i)|_1 = |\phi(i)|_2$ for all $i \in \{0, 1, 2\}$.

If $w$ cannot be de-substituted, then either $w$ has even length and occurs at an odd index in $\mathbf{x}$, or $w$ has odd length.
If $w$ has odd length, then deleting either the first or last letter results in a word that can be de-substituted, so $\big||w|_1 - |w|_2\big| \leq 1$.
If $w$ has even length and occurs at an odd index, then its first letter is $0$ or $2$ and its last letter is $0$ or $1$; deleting the first and last letters results in a word that can be de-substituted, so $\big||w|_1 - |w|_2\big| \leq 1$.

Finally, observe that if for all factors of a word $u$, the numbers of two letters $x$ and $y$ differ by at most $1$, then $x$ and $y$ alternate in $u$. 
\end{proof}

\begin{lemma}\label{lem:reversalPD}
Let $\tau$ be the morphism defined by $\tau : 0 \mapsto 0, 1 \mapsto 2, 2 \mapsto 1$.
If $w$ is a factor of $\mathbf{x}$, then $\tau(w)^\mathrm{R}$ is also a factor of $\mathbf{x}$.
\end{lemma}

\begin{proof}
We first prove by induction that $$\tau(\phi(2u1))^\mathrm{R} = \phi(\tau(12u)^\mathrm{R})$$
for every factor of the form $2u1$ of $\mathbf{x}$. 

One checks that this is true for $21$ and $2001$. If $2u1$ is a factor not equal to $21$ nor $2001$, then $u$ must contain a $2$ and we can write $2u1=2u'12u''1$ where $2u'1$ and $2u''1$ are factors of $\mathbf{x}$. By the induction hypothesis we have
\begin{align*}
\tau(\phi(2u1))^\mathrm{R}  &= \tau(\phi(2u'12u''1))^\mathrm{R}\\
        &= \tau(\phi(2u''1))^\mathrm{R}\tau(\phi(2u'1))^\mathrm{R}\\
        &= \phi(\tau(12u'')^\mathrm{R}) \phi(\tau(12u')^\mathrm{R})\\
        &= \phi(\tau(12u'12u'')^\mathrm{R})\\ 
        &= \phi(\tau(12u)^\mathrm{R}).       
\end{align*}

We now prove the lemma by induction on the length of $w$. One can check by hand that the lemma is true for $w$ of length at most $15$. Assume the lemma is true for every factor of length at most $n\geq 15$, and let $w$ be a factor of length $n+1$.
Then $w$ is a factor of $\phi(v)$ for some factor $v$ of $\mathbf{x}$ with $\frac{n+1}{2}\le|v|\le\frac{n+3}{2}$.

Since all factors of length $4$ contain a $1$ and a $2$, there exists a factor $u$ such that $v$ is a factor of $2u1$ and $|2u1|\le\frac{n+3}{2}+6$. In particular, $w$ is a factor of $\phi(2u1)$ and $\tau(w)^\mathrm{R}$ is a factor of $\tau(\phi(2u1))^\mathrm{R}$. To obtain the conclusion, we just need to show that $\tau(\phi(2u1))^\mathrm{R}$ is a factor of $\mathbf{x}$.

As by Lemma~\ref{lem:2facPD}, a $2$ is always preceded by a $1$ in $\mathbf{x}$, the word $12u$ is a factor of $\mathbf{x}$ and it has length $|12u|\le\frac{n+3}{2}+6\le n$. By induction hypothesis, $\tau(12u)^\mathrm{R}$ is a factor of $\mathbf{x}$. Hence $\phi(\tau(12u)^\mathrm{R})$ is also a factor. Finally, using the previous result, $\tau(\phi(2u1))^\mathrm{R}=\phi(\tau(12u)^\mathrm{R})$ is a factor of $\mathbf{x}$.
\end{proof}

We can now express $\mathcal{P}^{(1)}_{\mathbf{x}}$ in terms of $\Delta_0$.

\begin{proof}[Proof of Proposition~\ref{prop:deltatoab_PD}]
Let $w$ be a factor of $\mathbf{x}$ of length $|w| = n$.

If $|w| - |w|_0 = |w|_1 + |w|_2$ is even, it follows from Lemma~\ref{lem:balancedPD} that $|w|_1 = |w|_2$.
Therefore every factor of length $n$ containing exactly $|w|_0$ zeros is abelian-equivalent to $w$, so the pair $(n, |w|_0)$ determines a unique abelian equivalence class of factors.

If $|w| - |w|_0$ is odd, then by Lemma~\ref{lem:balancedPD} either $|w|_1 = |w|_2 + 1$ or $|w|_2 = |w|_1 + 1$.
By Lemma~\ref{lem:reversalPD}, there is another factor, $v = \tau(w)^\mathrm{R}$, of length $n$, with $|v|_0 = |w|_0$ and $|v|_1 - |v|_2 = |w|_2 - |w|_1$.
Therefore both possibilities occur, so the number of abelian equivalence classes corresponding to a pair $(n,|w|_0)$ is $2$.

There are $\Delta_0(n)$+1 possible values for the number of $0$'s in a factor of length $n$.
Since each value occurs for some factor, we have
\begin{align*}
	\mathcal{P}^{(1)}_{\mathbf{x}}(n)
	&= \sum_{i = \m_0(n)}^{\M_0(n)}
		\begin{cases}
			1	& \text{if $n - i$ is even} \\
			2	& \text{if $n - i$ is odd}
		\end{cases} \\
	&= \sum_{j = n - \M_0(n)}^{n - \m_0(n)}
		\begin{cases}
			1	& \text{if $j$ is even} \\
			2	& \text{if $j$ is odd}.
		\end{cases}
\end{align*}
Therefore $\mathcal{P}^{(1)}_{\mathbf{x}}(n) = \frac{3}{2} \Delta_0(n) + c(n)$, where $c(n)$ depends only on the parities of $\Delta_0(n)$ and $n - \m_0(n)$; computing four explicit values allows one to determine the values of $c(n)$ and obtain the equation claimed for $\mathcal{P}^{(1)}_{\mathbf{x}}(n)$.
\end{proof}


\subsection{$\Delta_0(n)_{n\ge 0}$ is $2$-regular, $(\m_0(n)\bmod{2})_{n\ge 0}$ is $2$-automatic}\label{sec:52}

In this section, we prove the following result.

\begin{proposition}\label{prop:recdelta_PD}
Let $\ell\geq 2$ and $r$ such that $0\leq r < 2^{\ell}$. We have
  $$\Delta_0(2^\ell+r)=
  \begin{cases}
    \Delta_0(r)+2 & \text{if } r\leq 2^{\ell-1} \\
    \Delta_0(2^{\ell+1}-r) & \text{if } r>2^{\ell-1}.\\
   \end{cases}
$$
Moreover, 
$$\m_0(2^\ell+r)\equiv
  \begin{cases}
    \m_0(r) \pmod{2} & \text{if } r\leq 2^{\ell-1} \\
    \m_0(2^{\ell+1}-r)+\Delta_0(2^{\ell+1}-r) \pmod{2}&\text{if } r>2^{\ell-1}.\\
   \end{cases}
$$
\end{proposition}

Before giving the proof, we prove a corollary.
The $2$-regularity of $\mathcal{P}^{(1)}_{\mathbf{x}}(n)_{n\ge0}$ follows from Proposition~\ref{prop:deltatoab_PD} and Corollary~\ref{cor:2regDm_PD}.

\begin{corollary}\label{cor:2regDm_PD} The following statements are true.
    \begin{itemize}
      \item The sequence $\Delta_0(n)_{n\ge 0}$ is $2$-regular.
      \item The sequence $(\Delta_0(n)\bmod{2})_{n\ge 0}$ is $2$-automatic.
      \item The sequence $(\m_0(n)\bmod{2})_{n\ge 0}$ is $2$-automatic.
    \end{itemize}
\end{corollary}

\begin{proof} The first assertion is a direct consequence of Proposition~\ref{prop:recdelta_PD} and Theorem~\ref{thm:reflection_recurrence}.
Note that one can obtain explicit relations satisfied by $\Delta_0(n)_{n\ge 0}$ from Example~\ref{exa:l0=2}.
The second assertion follows from Lemma~\ref{lem:regmod}. 

For the last assertion, for $i\in\{0,\ldots,31\}$ we prove that, modulo $2$,

\[
	\m_0(32n+i) \equiv \begin{cases}
	\m_0(8n+1)  & \text{if } i\in\{1,5,9,17,25\}\\
	\m_0(8n+3)  & \text{if } i=11\\
	\m_0(8n+5)  & \text{if } i=21\\
	\m_0(8n+7)  & \text{if } i\in\{7,15,23,27,31\}\\
	 0 & \text{otherwise}
\end{cases}
\]
and
\[
	\Delta_0(32n+i) \equiv \begin{cases}
	\Delta_0(8n+1)  & \text{if } i\in\{1,5,9,17,25\}\\
	\Delta_0(8n+3)  & \text{if } i=11\\
	\Delta_0(8n+5)  & \text{if } i=21\\
	\Delta_0(8n+7)  & \text{if } i\in\{7,15,23,27,31\}\\
	 0 & \text{otherwise.}
\end{cases}
\]

By Lemma~\ref{lem:minmaxeven}, we already know that $\m_0(2n)\equiv \Delta_0(2n)\equiv 0 \pmod{2}$ for any $n\in\mathbb{N}$. Hence the relations above are true for $i$ even.
We prove the other relations by induction on $n$. They are true for $n=0$.  Let $n>0$ and assume the relations are satisfied for all $n'$ such that $0\le n'<n$. We can write $n=2^\ell+r$ with $\ell\geq 0$ and $0\leq r < 2^\ell$. Let $i\in\{1,\ldots,31\}$ be odd.

Assume first that $r<2^{\ell-1}$. We have $32n+i=2^{\ell+5}+32r+i$ and $32r+i<2^{\ell+4}$.
\begin{align*}
\m_0(32n+i)&\equiv \m_0(32r+i) \tag{Proposition~\ref{prop:recdelta_PD}}\\
&\equiv \m_0(8r+j) \tag{induction}\\
&\equiv \m_0(2^{\ell+3}+8r+j) \tag{Proposition~\ref{prop:recdelta_PD}}\\
&\equiv \m_0(8n+j) \pmod{2}
\end{align*}
for some $j\in\{0,\ldots,7\}$ according to the relations. 
A similar reasoning holds for the $\Delta_0$ relations.

Assume now that $r\geq 2^{\ell-1}$. Since $32r+i>2^{\ell+4}$, we have 
\begin{align*}
\m_0(32n+i)&\equiv \m_0(2^{\ell+6}-32r-i) + \Delta_0(2^{\ell+6}-32r-i) \tag{Proposition~\ref{prop:recdelta_PD}}\\
&\equiv \m_0(32n'+j) + \Delta_0(32n'+j)\pmod{2}
\end{align*}
with $j=32-i$ and $n'=2^{\ell+1}-r-1$. If $i\in\{3,13,19,29\}$, then $j\in\{3,13,19,29\}$. By the induction hypothesis, $\m_0(32n'+j)\equiv \Delta_0(32n'+j)\equiv0\pmod{2}$ and we are done.

For the remaining cases, $i,j\not\in\{3,13,19,29\}$. As $\m_0$ and $\Delta_0$ satisfy the same recurrence relations, by the induction hypothesis, there exists $k\in\{1,3,5,7\}$ such that 
\begin{align*}
\m_0(32n+i)&\equiv \m_0(8n'+k) + \Delta_0(8n'+k)\\
&\equiv \m_0(2^{\ell+4}-(8r+8-k))+\Delta_0(2^{\ell+4}-(8r+8-k))\\
&\equiv \m_0(2^{\ell+3}+(8r+8-k))\tag{{Proposition~\ref{prop:recdelta_PD}}}\\
&\equiv \m_0(8n+(8-k))\pmod{2}.
\end{align*}
Observe that the value of $8-k$ is the value given in the relation for $i$. This concludes the proof of the $\m_0$ relations. A similar argument works for the $\Delta_0$ relations.
\end{proof}

We break the proof of Proposition~\ref{prop:recdelta_PD} into three parts, covered by Lemmas~\ref{lem:power2PD},~\ref{lem:recminmax} and \ref{lem:recminmax2}. We first deal with powers of 2.

\begin{lemma}\label{lem:power2PD}
  Let $\ell \in \mathbb N$, $\ell\geq 1$. We have $\mathcal{P}^{(1)}_{\mathbf{x}}(2^\ell)=4$,
  \[\quad \Delta_0(2^\ell)=2,\ \M_0(2^{\ell+1})=2^\ell-\m_0(2^\ell) \text{ and }\ \m_0(2^{\ell+1})=2^\ell-\M_0(2^\ell).\]
\end{lemma}

\begin{proof}
Recall that $\Psi(w) = (|w|_0, |w|_1, |w|_2)$ is the Parikh vector of $w$.
We show by induction that 
\begin{multline*}
\{\Psi(w) : w\text{ factor of }\mathbf{x}\text{ with }|w|=2^\ell\}\\
= \{P_\ell +(0,0,0),P_\ell +(-2,1,1),P_\ell +(-1,1,0),P_\ell +(-1,0,1)\}
\end{multline*}
and that
\begin{align*}
	\Psi(\phi^\ell(0)) &=
	\begin{cases}
		P_\ell			& \text{if $\ell$ is even}\\
		P_\ell +(-2,1,1)	& \text{if $\ell$ is odd}
	\end{cases} \\
	\Psi(\phi^\ell(2)) &=
	\begin{cases}
		P_\ell +(-2,1,1)	& \text{if $\ell$ is even}\\
		P_\ell			& \text{if $\ell$ is odd,}
	\end{cases}
\end{align*}
where $P_\ell=(\frac{2^\ell+4}{3},\frac{2^\ell-2}{3},\frac{2^\ell-2}{3})$ if $\ell$ is odd and $P_\ell=(\frac{2^\ell+2}{3},\frac{2^\ell-1}{3},\frac{2^\ell-1}{3})$ if $\ell$ is even.
Since Parikh vectors of factors of length $2^\ell$ can take exactly four values, the conclusion is immediate. 

The result is true for $\ell\in\{1,2\}$. Let $\ell>2$ and assume the result holds for $\ell-1$. Let $w$ be a factor of length $2^\ell$. 

If $w$ can be de-substituted, then $w=\phi(v)$ for some factor $v$ of length $2^{\ell-1}$, and $\Psi(w)=(2|v|_2,|v|_0+|v|_1,|v|_0+|v|_1)$. Using the induction hypothesis, it is easy to check that $\Psi(w)=P_{\ell}$ or $\Psi(w)=P_{\ell}+(-2,1,1)$ and that the equalities for $\Psi(\phi^\ell(0)),\Psi(\phi^\ell(2))$ are satisfied.

If $w$ cannot be de-substituted, then $w$ occurs at an odd index in $\mathbf{x}$ and $w$ is of the form $$0^{-1}\phi(v)0, \quad 1^{-1}\phi(v)1, \quad 0^{-1}\phi(v)1 \quad \text{ or } \quad 1^{-1}\phi(v)0$$ for some factor $v$ of length $2^{\ell-1}$. If $w$ is of one of the first two forms, then $\Psi(w)=\Psi(\phi(v))$ and $\Psi(w)=P_{\ell}$ or $\Psi(w)=P_{\ell}+(-2,1,1)$ (as in the previous case).

If $w=0^{-1}\phi(v)1$, then $w$ can also be written as $w=0\phi(u)2^{-1}$ for some factor $u$ of length $2^{\ell-1}$. So both Parikh vectors $\Psi(\phi(v))$ and $\Psi(\phi(u))$ belong to $\{P_{\ell},P_{\ell}+(-2,1,1)\}$. Since by construction $\phi(v)$ has two more zeros than $\phi(u)$, we obtain $\Psi(\phi(v))=P_\ell$ and $\Psi(\phi(u))=P_\ell+(-2,1,1)$. Thus $\Psi(w)=\Psi(\phi(v))+(-1,1,0)=P_\ell+(-1,1,0)$.

Similarly, if $w=1^{-1}\phi(v)0$, then $\Psi(w)=P_\ell+(-1,0,1)$.

To conclude the proof, we just need to show that these four cases actually occur for all $\ell$. 
Since $\{\Psi(\phi^\ell(0)),\Psi(\phi^\ell(2))\}=\{P_\ell,P_\ell+(-2,1,1)\}$, consider all factors of length $2^\ell$ occurring between two consecutive occurrences of $\Psi(\phi^\ell(0))$ and $\Psi(\phi^\ell(2))$. By continuity\footnote{We mean by {\em continuity} that the number of $0$'s is varying by at most 1 between two factors of the same length starting at consecutive indexes.}, of the number of $0$'s, one of these factors must have a Parikh vector equal to $P_\ell+(-1,1,0)$ or $P_\ell+(-1,0,1)$. Using Lemma~\ref{lem:reversalPD}, we obtain that
$w$ is a factor of length $2^\ell$ with $\Psi(w)=P_\ell+(-1,1,0)$ if and only if 
$\tau(w)^\mathrm{R}$ is a factor of length $2^\ell$  with $\Psi(w)=P_\ell+(-1,0,1)$. So all four values actually occur. 
\end{proof}

To show Lemmas~\ref{lem:recminmax} and~\ref{lem:recminmax2}, we first prove the following technical result.

\begin{lemma}\label{lem:max2max0}
Let $u$ be a factor of $\mathbf{x}$ of length $n\ge1$. Let $\M_2(n)$ (resp.\ $\m_2(n)$) denote the maximum (resp.\ minimum) of $\{|w|_2: w\text{ factor of }\mathbf{x}\text{ of length }n\}$. We have $|u|_2=\M_2(n)$ if and only if $|\phi(u)|_0=\M_0(2n)$, and $|u|_2=\m_2(n)$ if and only if $|\phi(u)|_0=\m_0(2n)$.
\end{lemma}

\begin{proof}
For the first assertion, assume that $|u|_2=\M_2(n)$ and suppose that $|\phi(u)|_0<\M_0(2n)$. Note that $|\phi(u)|_0=2|u|_2$ by definition of $\phi$. Let $v$ be a factor of length $2n$ such that $|v|_0=\M_0(2n)$, which is even by Lemma~\ref{lem:minmaxeven}. In addition, we can assume that $v$ starts with $00$. Indeed, if it is not the case, then either $v$ starts with $01$ and ends with $0$, or $v$ is of the form $t00s$ where $t$ does not contain any zero. In the first case, we can consider the word $0v0^{-1}$ that starts with $00$ and has $\M_0(2n)$ zeros. In the second case, we can consider the word $00sw$ for some $w$ with $|w|=|t|$. This factor has also $\M_0(2n)$ zeros. Therefore $v$ can be de-substituted. So $v=\phi(z)$ and $|z|_2=\frac{1}{2}|v|_0 > |u|_2$, which is a contradiction.

For the other direction, assume $|\phi(u)|_0=\M_0(2n)$ and suppose $|u|_2$ does not maximize the number of $2$'s. Then there exists a factor $v$ of length $n$ such that $|v|_2=\M_2(n)$. Hence,
$$|\phi(v)|_0=2|v|_2 >2|u|_2=|\phi(u)|_0=\M_0(2n),$$
which is a contradiction. Similar arguments hold for the second assertion. 
\end{proof}

\begin{lemma}\label{lem:recminmax}
If $\ell\geq2$ and $0\leq r\leq 2^{\ell-1}$, then 
\begin{align*}
	\M_0(2^\ell+r) &= \M_0(2^\ell)+\M_0(r), \\
	\m_0(2^\ell+r) &= \m_0(2^\ell)+\m_0(r).
\end{align*}
\end{lemma}

\begin{proof}
We work by induction on $\ell$. One checks the case $\ell=2$. Let $\ell>2$ and assume the statements are true for $\ell-1$. Let $r$ such that $0\leq r\leq 2^{\ell-1}$.

Assume first that $r$ is even. We shall exhibit a factor of length $2^\ell+r$ that has $\M_0(2^\ell)+\M_0(r)$ zeros and maximizes the number of $0$'s. By the induction hypothesis, the result is true for $2^{\ell-1}+r/2$. So there exists a factor $u$ of length $2^{\ell-1}+r/2$ with $\m_0(2^{\ell-1}+r/2)=\m_0(2^{\ell-1})+\m_0(r/2)$ zeros.  In addition, we can assume that $u$ maximizes the number of $2$'s. Indeed, since $|u|_0=\m_0(2^{\ell-1}+r/2)$, $|u|_1+|u|_2$ is maximal among all factors of length $2^{\ell-1}+r/2$. If the number of $1$ and $2$ in $u$ is even, then $|u|_2=|u|_1$ is maximal. Otherwise, either $|u|_2=|u|_1+1$ and $|u|_2$ is maximal, or $|u|_2=|u|_1-1$ and $u$ does not maximize the number of $2$'s. In the last case, by Lemma~\ref{lem:reversalPD}, we can consider the factor $\tau(u)^\mathrm{R}$ which satisfies $|\tau(u)^\mathrm{R}|_0=|u|_0$ and $|\tau(u)^\mathrm{R}|_2=|u|_1$. Hence, $\tau(u)^\mathrm{R}$ minimizes the number of $0$'s and maximizes the number of $2$'s.

Let us write $u=vw$ with $|v|=2^{\ell-1}$ and $|w|=r/2$. Then, as $|v|_0+|w|_0=|u|_0=\m_0(2^{\ell-1})+\m_0(r/2)$, the words $v$ and $w$ minimize the number of $0$'s for words of their respective lengths.
The word $v$ maximizes also the number of $2$'s for factors of length $2^{\ell-1}$ because $|v|$ and $|v|_0=\m_0(2^{\ell-1})$ are even by Lemma~\ref{lem:power2PD} and so is $|v|_1+|v|_2$. Since $u$ maximizes the number of $2$'s and $|v|_2=|v|_1$, the word $w$ also maximizes the number of $2$'s. Hence, by Lemma~\ref{lem:max2max0}, $\phi(u)$, $\phi(v)$ and $\phi(w)$ maximize the number of $0$'s for words of their respective lengths. Thus,
$$\M_0(2^\ell+r)=|\phi(u)|_0=|\phi(v)|_0+|\phi(w)|_0=\M_0(2^\ell)+\M_0(r).$$

If $r$ is odd, we still have $0\leq r-1 \leq r+1 \leq 2^{\ell-1}$ and we can use the previous results:
\begin{align*}
	\M_0(2^\ell+r-1) &= \M_0(2^\ell)+\M_0(r-1), \\
	\M_0(2^\ell+r+1) &= \M_0(2^\ell)+\M_0(r+1).
\end{align*}
Note that $\M_0$ is even for even values and can only grow by 0 or 1. So there are two cases to consider: either $\M_0(2^\ell+r+1)=\M_0(2^\ell+r-1)$ or $\M_0(2^\ell+r+1)=\M_0(2^\ell+r-1)+2$.

If the two maxima are equal, then $\M_0(r+1)=\M_0(r-1)$, $\M_0(2^\ell+r)=\M_0(2^\ell+r-1)$ and $\M_0(r)=\M_0(r-1)$, and we are done. Otherwise, the two maxima differ by $2$, and then $\M_0(r+1)=\M_0(r-1)+2$, $\M_0(2^\ell+r)=\M_0(2^\ell+r-1)+1$ and $\M_0(r)=\M_0(r-1)+1$, and we are done.

A similar proof shows that $\m_0(2^\ell+r) = \m_0(2^\ell)+\m_0(r)$.
\end{proof}

Lemma~\ref{lem:recminmax2} will follow directly from the following lemma.

\begin{lemma}
If $\ell\geq 2$ and $2^{\ell-1}\leq r \leq 2^\ell$, then
\begin{align*}
	\M_0(2^{\ell+1})&=\M_0(2^{\ell}+r)+\m_0(2^{\ell}-r),\\
	\m_0(2^{\ell+1})&=\m_0(2^\ell+r)+\M_0(2^\ell-r).
\end{align*}
Moreover, there is a factor of length $2^{\ell+1}$ maximizing (resp.\ minimizing) the number of $0$'s such that the prefix of length $2^\ell+r$ also maximizes (resp.\ minimizes) the number of $0$'s. In addition, the first equality $\M_0(2^{\ell+1})=\M_0(2^{\ell}+r)+\m_0(2^{\ell}-r)$ holds even if $\ell=1$.
\end{lemma}

\begin{proof}
We proceed by induction on $\ell$. One checks that the results are true for $\ell=2$ and, for the first equality, for $\ell=1$. Let $\ell>2$ and assume both equalities hold for $\ell-1$. Let $r$ such that $2^{\ell-1}\le r\le 2^{\ell}$.

Assume first that $r$ is even. By the induction hypothesis, there exists a factor $u=vw$ of length $2^{\ell}$ such that
$$|u|_0=\m_0(2^{\ell})=\m_0(2^{\ell-1}+r/2)+\M_0(2^{\ell-1}-r/2),$$
$|v|=2^{\ell-1}+r/2$ and $v$ minimizes the number of $0$'s. Hence, $|v|_0=\m_0(2^{\ell-1}+r/2)$ and $|w|_0=\M_0(2^{\ell-1}-r/2)$.

Observe that $u$ maximizes the number of $2$'s as $|u|$ and $|u|_0=\m_0(2^{\ell})$ are even. In addition, we can assume that $v$ also maximizes the number of $2$'s. Indeed, if $v$ is of even length,  $|v|_0=\m_0(2^{\ell-1}+r/2)$ implies  $|v|_2$ is maximal. If $v$ is of odd length and $v$ does not maximize the number of $2$'s, then it ends with $1$. Thus, $v$ is followed by a $2$. In particular, $v$ occurs at an even index in $\mathbf{x}$. So is $u$ and $u12$ or $u00$ is a factor of $\mathbf{x}$. If $u12$ is a factor, then consider, instead of $u$, $u'=z^{-1}u1$ where $z$ denotes the first letter of $u$. In that case, the prefix of length $2^{\ell-1}+r/2$ of $u'$ is $z^{-1}v2$. It still minimizes the number of $0$'s and now maximizes the number of $2$'s. Assume now that $u00$ is a factor. Observe that $\mathbf{x}$ is the fixed point of $\phi$. So it is also the fixed point of $\phi^2$. Therefore, $\mathbf{x}$ is a concatenation of blocks of length $4$ of the form $\phi^2(0)=\phi^2(1)=1200$ and $\phi^2(2)=1212$. Since $u00$ is a factor of $\mathbf{x}$, the only extension of this factor is $12u00$ as $|u|=2^\ell\equiv 0\pmod{4}$. Consider then $u'=2u2^{-1}$.

Since $|u|_1=|u|_2$ and $|v|_2\ge |v|_1$, $|w|_1\ge|w|_2$. Thus, as $|w|_0=\M_0(2^{\ell-1}-r/2)$, $w$ minimizes the number of $2$'s. By Lemma~\ref{lem:max2max0}, we obtain $|\phi(u)|_0=\M_0(2^{\ell+1})$, $|\phi(v)|_0=\M_0(2^\ell+r)$, $|\phi(w)|_0=\m_0(2^\ell-r)$. So
\begin{align*}
	\M_0(2^{\ell+1})&=|\phi(u)|_0=|\phi(v)|_0+|\phi(w)|_0 \\
	&=\M_0(2^\ell+r)+\m_0(2^\ell-r).
\end{align*}

We can show similarly that $\m_0(2^{\ell+1})=\m_0(2^\ell+r)+\M_0(2^\ell-r)$. Note that in this case, we can assume that the factor $u$ with $|u|_0=\M_0(2^\ell)$, given by the induction hypothesis, starts with $00$ as in the proof of Lemma~\ref{lem:max2max0}.

Assume now that $r$ is odd. Then $2^{\ell-1}\le r-1<r+1\le 2^{\ell}$ and we can apply the previous result:
\begin{align*}
	\M_0(2^{\ell+1})&=\M_0(2^{\ell}+r-1)+\m_0(2^{\ell}-r+1)\\
	&=\M_0(2^{\ell}+r+1)+\m_0(2^{\ell}-r-1).
\end{align*}
Since $\M_0$ is even for even values and can only grow by 0 or 1, there are two cases to consider: either $\M_0(2^\ell+r-1)=\M_0(2^\ell+r+1)$ or $\M_0(2^\ell+r-1)+2=\M_0(2^\ell+r+1)$.

If the two maxima are equal, then $\m_0(2^\ell-r+1)=\m_0(2^\ell-r-1)=\m_0(2^\ell-r)$ and $\M_0(2^{\ell}+r)=\M_0(2^{\ell}+r-1)$, and we are done. Otherwise, the two maxima differ by $2$, and then $\m_0(2^\ell-r+1)-2=\m_0(2^\ell-r-1)$. So $\M_0(2^{\ell}+r)=\M_0(2^{\ell}+r-1)+1$ and $\m_0(2^\ell-r)=\m_0(2^\ell-r+1)-1$, and we are done.
Using similar argument, we can conclude that $\m_0(2^{\ell+1})=\m_0(2^\ell+r)+\M_0(2^\ell-r)$.

For the construction of the factors, one can construct them using the factors $\phi(u)$ and $\phi(u')$ given for $r-1$ and $r+1$ in the previous construction. We consider the same two cases as before.

If the maxima are equal, then $\M_0(2^{\ell}+r)=\M_0(2^{\ell}+r-1)$. By construction, $\phi(u)$ has a prefix $\phi(v)$ of length $2^\ell+r-1$, maximizing the number of $0$'s. The letter $z$ following the prefix $\phi(v)$ in $\phi(u)$ is not a $0$. Otherwise, $\phi(v)0$ would be a factor of length $2^\ell+r$ with $\M_0(2^{\ell}+r)+1$ zeros, which is a contradiction. Hence, $\phi(v)z$ is a prefix of length $2^\ell+r$ of $\phi(u)$ that maximizes the number of $0$'s.

If $\M_0(2^\ell+r-1)+2=\M_0(2^\ell+r+1)$, then  $\M_0(2^{\ell}+r)=\M_0(2^{\ell}+r+1)-1$. By construction, $\phi(u')$ has a prefix $\phi(v')$ of length $2^\ell+r+1$, maximizing the number of $0$'s. This prefix must end with $0$. Otherwise, deleting the last letter of $\phi(v')$ would give a factor of length $2^\ell+r$ with $\M_0(2^{\ell}+r+1)=\M_0(2^{\ell}+r)+1$ zeros, which is a contradiction. Hence, $\phi(v')0^{-1}$ is a prefix of length $2^\ell+r$ of $\phi(u')$ that maximizes the number of $0$'s.

A similar construction yields a factor of length $2^{\ell+1}$ minimizing the number of $0$'s such that the prefix of length $2^\ell+r$ also minimizes the number of $0$'s.
\end{proof}

The previous lemma permits us to reformulate some relations between the two sequences $\M_0(n)_{n\ge0}$ and $\m_0(n)_{n\ge0}$.

\begin{lemma}\label{lem:recminmax2}
If $\ell\geq 2$ and $2^{\ell-1}\leq r \leq 2^\ell$, then
\begin{align*}
\M_0(2^{\ell}+r)&=2^{\ell} -\m_0(2^{\ell+1}-r),\\
\m_0(2^{\ell}+r)&=2^\ell-\M_0(2^{\ell+1}-r).
\end{align*}
The first equality holds even if $\ell=1$.
\end{lemma}

\begin{proof}
One can check the first equality for $\ell=1$. Let $\ell\geq 2$ and $r$ such that $2^{\ell-1}\leq r \leq 2^\ell$. From the previous lemma, we have
$$\M_0(2^{\ell}+r)=\M_0(2^{\ell+1})-\m_0(2^{\ell}-r).$$
Note that, by Lemma~\ref{lem:power2PD}, we have $\M_0(2^{\ell+1})=2^\ell-\m_0(2^\ell)$. Moreover, by Lemma~\ref{lem:recminmax}, since $0\le 2^{\ell}-r\le 2^{\ell}$, we get
$$\m_0(2^\ell)+\m_0(2^\ell-r)=\m_0(2^{\ell}+2^{\ell}-r).$$
Since similar relations hold when exchanging $\m_0$ and $\M_0$, the conclusion follows.
\end{proof}

The proof of Proposition~\ref{prop:recdelta_PD} about the reflection relation satisfied by $\Delta_{0}(n)$ and the recurrence relation of $\m_{0}(n)$ is now immediate.

\begin{proof}[Proof of Proposition~\ref{prop:recdelta_PD}]
Let $\ell\geq 2$. For $r$ such that $0\leq r\leq 2^{\ell-1}$, subtracting the two relations provided by Lemma~\ref{lem:recminmax} gives $\Delta_{0}(2^\ell+r)=\Delta_0(2^\ell)+\Delta_{0}(r)$ and we can conclude using the first relation given in Lemma~\ref{lem:power2PD}, $\Delta_0(2^\ell)=2$. Furthermore, $\m_0(2^\ell+r)\equiv \m_0(2^\ell)+\m_0(r) \pmod{2}$ by Lemma~\ref{lem:recminmax}. The expression for $\m_0(2^\ell+r)$ follows since $\m_0(2^\ell)\equiv 0 \pmod{2}$ by Lemma~\ref{lem:power2PD}.

For $2^{\ell-1}< r < 2^\ell$, subtracting the two relations provided by Lemma~\ref{lem:recminmax2} permits us to conclude the proof of the expression claimed for $\Delta_0(2^\ell+r)$. Moreover, using  Lemma~\ref{lem:recminmax2}, we get
\begin{align*}
\m_0(2^\ell+r)&\equiv \M_0(2^{\ell+1}-r) \pmod{2}\\
&\equiv  \m_0(2^{\ell+1}-r)+\Delta_0(2^{\ell+1}-r) \pmod{2}. \qedhere
\end{align*}
\end{proof}


\subsection{Another proof of the $2$-regularity of $\mathcal{P}^{(1)}_{\mathbf{x}}(n)_{n\ge 0}$}\label{sec:53} 

In this section we prove the $2$-regularity of the abelian complexity $\mathcal{P}^{(1)}_{\mathbf{x}}(n)_{n\ge 0}$ in a second way, by proving Theorem~\ref{thm:recab_PD}.
The proof makes use of Propositions~\ref{prop:deltatoab_PD} and \ref{prop:recdelta_PD}.

\begin{proof}[Proof of Theorem~\ref{thm:recab_PD}]
If $2^{\ell-1}\leq r \leq 2^\ell$, since all the conditions in Proposition~\ref{prop:deltatoab_PD} are equivalent whether considering $2^\ell+r$ or $2^{\ell+1}-r$, we have $$\mathcal{P}^{(1)}_{\mathbf{x}}(2^\ell+r)=\mathcal{P}^{(1)}_{\mathbf{x}}(2^{\ell+1}-r).$$

Assume now that $0\leq r \leq 2^{\ell-1}$.
If $\Delta_0(2^\ell+r)$ is odd, $\Delta_0(r)$ is also odd by Proposition~\ref{prop:recdelta_PD}. By Proposition~\ref{prop:deltatoab_PD}, we have $\mathcal{P}^{(1)}_{\mathbf{x}}(2^\ell+r)=\frac{3}{2}(\Delta_{0}(2^\ell+r)+1)$ and $\mathcal{P}^{(1)}_{\mathbf{x}}(r)=\frac{3}{2}(\Delta_{0}(r)+1)$. By Proposition~\ref{prop:recdelta_PD}, we have $\Delta_{0}(2^\ell+r)=\Delta_{0}(r)+2$. Putting these three equalities together, we get $\mathcal{P}^{(1)}_{\mathbf{x}}(2^\ell+r)=\mathcal{P}^{(1)}_{\mathbf{x}}(r)+3$.

The other cases can be done similarly.
If $\Delta_{0}(2^\ell+r)$ and $2^{\ell}+r-\m_{0}(2^\ell+r)$ are even, then $\Delta_0(r)$ and $r-\m_0(r)$ are even and
\begin{align*}
	\mathcal{P}^{(1)}_{\mathbf{x}}(2^\ell+r)
	&= \tfrac{3}{2}\Delta_{0}(2^\ell+r)+1	&\text{(by Proposition~\ref{prop:deltatoab_PD})}\phantom{.}\\
	&= \tfrac{3}{2}(\Delta_{0}(r)+2)+1	&\text{(by Proposition~\ref{prop:recdelta_PD})}\phantom{.}\\
	&= \mathcal{P}^{(1)}_{\mathbf{x}}(r)+3	&\text{(by Proposition~\ref{prop:deltatoab_PD})}.
\end{align*}

If $\Delta_{0}(2^\ell+r)$ is even and $2^{\ell}+r-\m_{0}(2^\ell+r)$ is odd, then $\Delta_0(r)$ is even and $r-\m_0(r)$ is odd. Then
\begin{align*}
	\mathcal{P}^{(1)}_{\mathbf{x}}(2^\ell+r)
	&= \tfrac{3}{2}\Delta_{0}(2^\ell+r)+2	& \text{(by Proposition~\ref{prop:deltatoab_PD})}\phantom{.}\\
	&= \tfrac{3}{2}(\Delta_{0}(r)+2)+2	& \text{(by Proposition~\ref{prop:recdelta_PD})}\phantom{.}\\
	&= \mathcal{P}^{(1)}_{\mathbf{x}}(r)+3	& \text{(by Proposition~\ref{prop:deltatoab_PD})}. & \qedhere
\end{align*}
\end{proof}

One can prove the following result in a manner similar to the proof of Theorem~\ref{thm:reflection_recurrence}.
There may be simpler recurrences, but these relations exhibit the same symmetry as in Theorem~\ref{thm:reflection_recurrence}.

\begin{theorem}
The abelian complexity sequence $\mathcal{P}^{(1)}_{\mathbf{x}}(n)_{n\ge0}$ of the $2$-block coding of the period-doubling word satisfies the following relations.
{\small
\begin{align*}
\mathcal{P}^{(1)}_{\mathbf{x}}(8n)&= \mathcal{P}^{(1)}_{\mathbf{x}}(2n) \\
4\mathcal{P}^{(1)}_{\mathbf{x}}(8n+1)&=-2\mathcal{P}^{(1)}_{\mathbf{x}}(2n+1)+7\mathcal{P}^{(1)}_{\mathbf{x}}(4n+1)-2\mathcal{P}^{(1)}_{\mathbf{x}}(4n+2)+\mathcal{P}^{(1)}_{\mathbf{x}}(4n+3)\\
4\mathcal{P}^{(1)}_{\mathbf{x}}(8n+2)&= -6\mathcal{P}^{(1)}_{\mathbf{x}}(2n+1)+9\mathcal{P}^{(1)}_{\mathbf{x}}(4n+1)-2\mathcal{P}^{(1)}_{\mathbf{x}}(4n+2)+3\mathcal{P}^{(1)}_{\mathbf{x}}(4n+3)\\
4 \mathcal{P}^{(1)}_{\mathbf{x}}(8n+3)&= -6\mathcal{P}^{(1)}_{\mathbf{x}}(2n+1)+5\mathcal{P}^{(1)}_{\mathbf{x}}(4n+1)+2\mathcal{P}^{(1)}_{\mathbf{x}}(4n+2)+3\mathcal{P}^{(1)}_{\mathbf{x}}(4n+3)\\
\mathcal{P}^{(1)}_{\mathbf{x}}(8n+4)&=\mathcal{P}^{(1)}_{\mathbf{x}}(4n+2)\\
4 \mathcal{P}^{(1)}_{\mathbf{x}}(8n+5)&= -6\mathcal{P}^{(1)}_{\mathbf{x}}(2n+1)+3\mathcal{P}^{(1)}_{\mathbf{x}}(4n+1)+2\mathcal{P}^{(1)}_{\mathbf{x}}(4n+2)+5\mathcal{P}^{(1)}_{\mathbf{x}}(4n+3)\\
4 \mathcal{P}^{(1)}_{\mathbf{x}}(8n+6)&= -6\mathcal{P}^{(1)}_{\mathbf{x}}(2n+1)+3\mathcal{P}^{(1)}_{\mathbf{x}}(4n+1)-2\mathcal{P}^{(1)}_{\mathbf{x}}(4n+2)+9 \mathcal{P}^{(1)}_{\mathbf{x}}(4n+3)\\
4\mathcal{P}^{(1)}_{\mathbf{x}}(8n+7)&=-2\mathcal{P}^{(1)}_{\mathbf{x}}(2n+1)+\mathcal{P}^{(1)}_{\mathbf{x}}(4n+1)-2\mathcal{P}^{(1)}_{\mathbf{x}}(4n+2)+7\mathcal{P}^{(1)}_{\mathbf{x}}(4n+3)
\end{align*}}
\end{theorem}

\section{$2$-abelian complexity of the period-doubling word}\label{sec:2ab_comp_PD}

To prove the $2$-regularity of $\mathcal{P}^{(2)}_{\mathbf{p}}(n)_{n\ge 0}$, the aim of this section is to express the $2$-abelian complexity $\mathcal{P}^{(2)}_{\mathbf{p}}$ in terms of the $1$-abelian complexity $\mathcal{P}^{(1)}_{\mathbf{x}}$ and the following additional $2$-regular functions.

\begin{definition}
    We define the {\em max-jump} function $\JM_0:\mathbb N \to \{0,1\}$ by $\JM_0(0)=0$ and, for $n\geq 1$,
\[
	\JM_0(n)=
	\begin{cases}
		1	& \text{if $\M_0(n)>\M_0(n-1)$} \\
		0	& \text{otherwise},
	\end{cases}
\]
i.e., $\JM_0(n)=1$ when the function $\M_0$ increases.  Similarly, let  $\jm_0:\mathbb N \to \{0,1\}$ be the {\em min-jump} function defined by
\[
	\jm_0(n)=
	\begin{cases}
		1	& \text{if $\m_0(n+1) > \m_0(n)$} \\
		0	& \text{otherwise}.
	\end{cases}
\]
\end{definition}

Since $\M_0(n)$ and $\m_0(n)$ are non-decreasing, we can write
\begin{align*} 
	\JM_0(n+1)&=\M_0(n+1)-\M_0(n), \\
	\jm_0(n)&=\m_0(n+1)-\m_0(n).
\end{align*}

The relationship between these sequences and $\mathcal{P}^{(2)}_{\mathbf{p}}$ and $\mathcal{P}^{(1)}_{\mathbf{x}}$ is stated in the following result.
\begin{proposition}\label{prop:abto2ab_PD} Let $n\geq 1$ be an integer.  Then
$$\mathcal{P}^{(2)}_{\mathbf{p}}(n+1)-\mathcal{P}^{(1)}_{\mathbf{x}}(n)=
  \begin{cases}
    0 & \text{if }  n\text{ is odd}\\
    \frac{\Delta_0(n)}{2}+ 1-\JM_0(n)-\jm_0(n) & \text{if } n\text{ is even.}\\
   \end{cases}
$$
\end{proposition}

We require several preliminary results.

\begin{proposition}
Let $u$ and $v$ be factors of $\mathbf{p}$ of length $n$. Let $u'$ and $v'$ be the $2$-block codings of $u$ and $v$. The factors $u$ and $v$ are $2$-abelian equivalent if and only if $u'$ and $v'$ are abelian equivalent and either $u'$ and $v'$ both start with $2$ or none of them start with $2$.
\end{proposition}

\begin{proof}
By Lemma~\ref{lem:abel}, $u$ and $v$ are $2$-abelian equivalent if and only if they start with the same letter and have the same number of factors $00$, $01$ and $10$. The number of $00$ (respectively $01$ and $10$) in $u$ is exactly the number of $0$ (resp.\ $1$ and $2$) in $u'$.  Moreover,  $u$ starts with $0$ (resp.\ by $1$) if and only if $u'$ starts with $0$ or $1$ (resp.\ by $2$). Therefore, $u$ and $v$ are $2$-abelian equivalent if and only if $u'$ and $v'$ are abelian equivalent and both start with $2$ or none of them start with $2$.\end{proof}

To compute $\mathcal{P}^{(2)}_{\mathbf{p}}$, we will use the abelian complexity of $\mathbf{x}=\blo(\mathbf{p},2)$, $\mathcal{P}^{(1)}_{\mathbf{x}}$, and study when an abelian equivalence class of $\mathbf x$ splits into two $2$-abelian equivalence classes of $\mathbf{p}$, or in other words, study when two abelian equivalent factors of $\mathbf x$ can start, respectively, with $2$ and with $0$ or $1$.
If the class does not split, we say that it leads to only one class.

\begin{lemma}\label{lem:diff12}
Let $\mathcal X$ be an abelian equivalence class of factors of length $n$ of $\mathbf x$. 
If the number of $1$'s in an element of $\mathcal X$ differs from the number of $2$'s, then $\mathcal X$ leads to only one $2$-abelian equivalence class of $\mathbf{p}$.
\end{lemma}

\begin{proof}
It is enough to prove that if an element of $\mathcal X$ starts with $2$, all the other elements of $\mathcal X$ start with $2$.
If $u$ starts with $2$, then all the elements of $\mathcal X$ have more $2$'s than $1$'s. But any factor with more $2$'s than $1$'s starts with a $2$.
\end{proof}

\begin{corollary}\label{cor:nodd}
If $n$ is odd, $\mathcal{P}^{(2)}_{\mathbf{p}}(n+1)=\mathcal{P}^{(1)}_{\mathbf{x}}(n)$.
\end{corollary}

\begin{proof}
Let $\mathcal X$ be an abelian equivalence class of factors of odd length $n$.
If no element of $\mathcal X$ starts with a $2$, $\mathcal X$ leads to only one $2$-abelian equivalence class of factors of $\mathbf{p}$.
So assume that there is a factor $u$ in $\mathcal X$ starting with $2$. Since $n$ is odd, we can write $u=2\phi(u')$. 
Then the number of $0$'s in $u$ is even and there is a different number of $2$'s than $1$'s. By Lemma~\ref{lem:diff12}, $\mathcal X$ again leads to a unique $2$-abelian equivalence class of $\mathbf p$.
\end{proof}

\begin{corollary}\label{cor:nevenn0odd}
Let $\mathcal X$  be an abelian equivalence class of factors of $\mathbf x$ of even length $n$ with an odd number of zeros. 
Then $\mathcal X$ leads to only one $2$-abelian equivalence class of $\mathbf{p}$.
\end{corollary}

\begin{proof}
Factors in $\mathcal X$ have an odd number of $1$'s and $2$'s counted together, so the number of $1$'s and the number of $2$'s are different and we can apply Lemma~\ref{lem:diff12}.
\end{proof}

Thus, an abelian equivalence class $\mathcal X$ of factors of length $n$ of $\mathbf x$ can possibly lead to two $2$-abelian equivalence classes of factors of length $n+1$ of $\mathbf {p}$ only if $n$ is even and if there are an even number of zeros in $\mathcal X$. In most cases $\mathcal X$ will indeed lead to two different equivalence classes. The exceptions are identified by the following lemma.

\begin{lemma}\label{lem:ext}
Let $n$ be a positive even integer and $n_0$ such that $\m_0(n)\leq n_0 \leq \M_0(n)$.
Let $\mathcal X$ be an abelian equivalence class of factors of $\mathbf x$ of length $n$ with exactly $n_0$ zeros.
\begin{itemize}
\item
We have $n_0=\M_0(n)$ and $\JM_0(n)=1$ if and only if every factor $u$ in $\mathcal X$ can be written as $u=00u'00$.
\item
We have $n_0=\m_0(n)$ and $\jm_0(n)=1$ if and only if every factor $u$ in $\mathcal X$ is preceded and followed only by $00$.
\end{itemize}
\end{lemma}

\begin{proof}
We start by proving the first part of the lemma.
Assume that all the elements of $\mathcal X$ have the form $00u'00$. 
In particular, $n_0$ is even. If $n_0\neq \M_0(n)$, it means that there is a factor $v$ of length $n$ with $n_0+1$ zeros. Indeed, sliding a window of length $n$ from a word of $\mathcal X$ to a factor with $\M_0(n)$ zeros gives factors with all possibilities between $n_0$ and $\M_0(n)$ for the number of zeros. Since $|v|_0$ is odd and $n$ is even, we must have $v=0\phi(v')1$ or $v=2\phi(v')0$. But then $0^{-1}v2$ or $1v0^{-1}$ is an element of $\mathcal X$ not of the form $00u'00$, a contradiction.
Hence $n_0=\M_0(n)$. If $\JM_0(n)=0$, then $\M_0(n-1)=n_0$ and there is a factor $v$ of odd length $n-1$ with even number $n_0$ of $0$'s. We must have $v=2\phi(v')$ or $v=\phi(v')1$ but then $1v$ or $v2$ is an element of $\mathcal X$ not of the form $00u'00$, a contradiction and $\JM_0(n)=1$.

For the other direction, assume that $n_0=\M_0(n)$ and $\JM_0(n)=1$. In particular, $\M_0(n-1)=n_0-1$. Assume there exists a factor $u$ of $\mathcal X$ not of the form $u=00u'00$. Since $u$ has even length and even number of $0$'s, we must have $u=01u'20$ or $u$ has its first or last letter $y$ not equal to $0$. In the first case, $v=001u'$ has length $n-1$ and $n_0$ zeros, a contradiction. In the second case, removing the letter $y$ leads also to a factor of length $n-1$ with $n_0$ zeros.

The second part of the lemma is similar. Assume first that all the elements of $\mathcal X$ are preceded and followed by $00$.  In particular, $n_0$ is even. If $n_0\neq \m_0(n)$, there is a factor $v$ of length $n$ with $n_0-1$ zeros. Since $|v|_0$ is odd but $n$ is even, we must have $v=0\phi(v')1$ or $v=2\phi(v')0$ but then $0v1^{-1}$ or $2^{-1}v0$ is an element of $\mathcal X$ that starts or ends with $00$ and so is preceded or followed by $12$, a contradiction.
Hence we have $n_0=\m_0(n)$. If $\jm_0(n)= 0$, then $\m_0(n+1)=n_0$ and there is a factor $v$ of odd length $n+1$ with even number $n_0$ of $0$'s. We must have $v=2\phi(v')$ or $v=\phi(v')1$ but then $\phi(v')$ is an element of $\mathcal X$ without a $00$ preceding or following it.

For the other direction, assume that $n_0=\m_0(n)$ and $\jm_0(n)=1$. In particular $\m_0(n+1)=n_0+1$. If there exists a factor $u$ of $\mathcal X$ such that $1u$, $2u$, $u1$ or $u2$ is a factor, then $\m_0(n+1)\leq n_0$, a contradiction. Hence all the factors $u$ of $\mathcal X$ can only be extended by $0u0$.
Finally, note that $u\in \mathcal X$ cannot occur in $\mathbf{x}$ at odd index. In other words, any $u\in \mathcal X$ can be de-substituted. Indeed, if it is not the case, then $u$ is of the form $0\phi(u')0$, $0\phi(u')1$, $2\phi(u')0$ or $2\phi(u')1$. If $u$ is of the first form, then $\phi(u')001$ is a factor of length $n+1$ with only $n_0$ zeros, which is a contradiction. Otherwise, $u$ is of one of the last three forms. Then either $u2$ or $1u$ is a factor of $\mathbf{x}$, which is not possible. So the only extension of $u$ as a factor of $\mathbf{x}$ is $00u00$.
\end{proof}

\begin{lemma}\label{lem:nevenn0even}
Let $n$ be a positive even integer and $n_0$ even such that $\m_0(n)\leq n_0 \leq \M_0(n)$.
Let $\mathcal X$ be an abelian equivalence class of factors of $\mathbf x$ of length $n$ with $n_0$ zeros.
The class $\mathcal X$ leads to only one $2$-abelian equivalence class of $\mathbf{p}$ if and only if $n_0=\m_0(n)$ and $\jm_0(n)=1$ or $n_0=\M_0(n)$ and $\JM_0(n)=1$.
Otherwise, $\mathcal X$ splits into two classes.
\end{lemma}

\begin{proof}
The factors in $\mathbf x$ of length $n = 2$ are $00,01,12,21,20$. The two classes to consider are $\mathcal X_1=\{00\}$, which leads to one class, and $\mathcal X_2=\{12,21\}$, which splits into two classes. Since $\JM_0(2)=1$ and $\jm_0(2)=0$, the proposition is true.

Hence let $n\geq 4$ even.
If $n_0=\m_0(n)$ and $\jm_0(n)=1$, then by Lemma~\ref{lem:ext}, all the elements of $\mathcal X$ are preceded by $00$. In particular, they all start with $1$ and $\mathcal X$ leads to only one $2$-abelian equivalence class. Similarly, if $n_0=\M_0(n)$ and $\JM_0(n)=1$, then by Lemma~\ref{lem:ext}, all the elements of $\mathcal X$ start with $0$ and we have only one class.

Assume now that $\mathcal X$ leads to only one class. 
If an element $u$ of $\mathcal X$ starts with $2$, we have $u=2\phi(u')1$ since $n$ and $n_0$ are even. Then $1u1^{-1}$ is an element of $\mathcal X$ starting with $1$ and $\mathcal X$ splits into two classes. Hence every element $u$ of $\mathcal X$ starts with $0$ or $1$. Assume there exists a factor $u$ in $\mathcal X$ that starts with a $1$. Then $u=12\phi(u')$ and $u$ cannot be followed by a $1$ since otherwise $1^{-1}u1$ would be an element of $\mathcal X$ starting with $2$. Hence $u$ is always followed by $00$ and so ends with $12$. Similarly, it can only be preceded by $00$. Hence all the factors in $\mathcal X$ starting with a $1$ are preceded and followed by  $00$. In particular, if a factor in $\mathcal X$ starts with $1$ and occurs in $\mathbf x$ at index $i$, then the two factors starting at indices $i - 1$ and $i + 1$ in $\mathbf x$ have $n_0+1$ zeros.
Assume now there exists a factor $u$ in $\mathcal X$ starting with a $0$. Then, $u$ can be de-substituted. Otherwise, as $n$ and $n_0$ are even, $u$ is of the form $0\phi(u')0$ where $\phi(u')$ ends with $12$. Thus $2\phi(u')2^{-1}$ is an element of $\mathcal X$ starting with $2$, which is a contradiction. Hence $u$ starts with $00$. If $u$ ends with $12$, then again, $2u2^{-1}$ is an element of $\mathcal X$ starting with $2$. Hence $u=00\phi(u')00$ and all elements of $\mathcal X$ starting with $0$ start and end with $00$. In particular, if a factor in $\mathcal X$ starts with $0$ and occurs in $\mathbf x$ at index $i$, then the two factors starting at indices $i - 1$ and $i + 1$ in $\mathbf x$ have $n_0-1$ zeros.

If no elements of $\mathcal X$ start with $1$ or no elements start with $0$, we are done by Lemma~\ref{lem:ext}. 
Otherwise, since one can show that $\mathbf x$ is uniformly recurrent\footnote{A word is \emph{uniformly recurrent} if every factor occurs infinitely often and, for each factor, there is a constant $c$ such that two consecutive occurrences of the factor occur within $c$ of each other. To prove that $\mathbf x$ is uniformly recurrent, it is enough to observe that $\phi$ is primitive since for each letter $y\in\{0,1,2\}$, $\phi^3(y)$ contains all the letters.}, we can assume that there exist a factor $u\in \mathcal X$ that starts with $0$ and occurs at index $i$ in $\mathbf x$, and a factor $v\in \mathcal X$ that starts with $1$ and occurs at index $i+\ell$ in $\mathbf x$, such that any factor $w_s$ of length $n$ occurring at index $i+s$ in $\mathbf x$ does not belong to $\mathcal{X}$ for $0<s<\ell$. Then $w_1$ has $n_0-1$ zeros whereas $w_{\ell-1}$ has $n_0+1$ zeros. But there is no factor $w_s$ with $n_0$ zeros. This is a contradiction since the number of $0$'s changes by at most one between two factors of the same length starting at consecutive indexes.  
\end{proof}

\begin{proof}[Proof of Proposition~\ref{prop:abto2ab_PD}]
The case $n$ odd is given by Corollary~\ref{cor:nodd}. Assume now that $n$ is even. Then by Lemma~\ref{lem:minmaxeven}, $\m_0(n)$ and $\M_0(n)$ are even, and therefore $\Delta_0(n)$ is even as well.
 Let $\mathcal X$ be an abelian equivalence class of factors of $\mathbf x$ of length $n$. Let $n_0$ be the number of $0$'s in the elements of $\mathcal X$. There are exactly $\frac{\Delta_0(n)}{2}$ odd values of $n_0$ and $\frac{\Delta_0(n)}{2}+1$ even values.
 By Corollary~\ref{cor:nevenn0odd}, if $n_0$ is odd, $\mathcal X$ leads to one $2$-abelian equivalence class of $\mathbf{p}$. 
 By Lemma~\ref{lem:nevenn0even}, $\mathcal X$ splits into two classes except for $n_0=\m_0(n)$ if $\jm_0(n)=1$ and for $n_0=\M_0(n)$ if $\JM_0(n)=1$. Hence there are in total $\frac{\Delta_0(n)}{2}+1 -\JM_0(n)-\jm_0(n)$ cases where $\mathcal X$ leads to two $2$-abelian equivalence classes of $\mathbf{p}$ instead of one and this is exactly the difference between $\mathcal{P}^{(2)}_{\mathbf{p}}(n+1)$ and $\mathcal{P}^{(1)}_{\mathbf{x}}(n)$.
\end{proof}

\begin{corollary}
The sequence $\mathcal{P}^{(2)}_{\mathbf{p}}(n)_{n\ge 0}$ is $2$-regular.
\end{corollary}

\begin{proof} We can make use of Lemma~\ref{lem:compo}. 
Thanks to Proposition~\ref{prop:abto2ab_PD}, $\mathcal{P}^{(2)}_{\mathbf{p}}(n+1)$ can be expressed as a combination of $\mathcal{P}^{(1)}_{\mathbf{x}}(n)$, $\Delta_0(n)$, $\JM_0(n)$, $\jm_0(n)$ using the predicate $(n\bmod{2})$.
Note that the predicate $(n\bmod{2})$ is trivially $2$-automatic.

We proved the $2$-regularity of $\mathcal{P}^{(1)}_{\mathbf{x}}(n)_{n \geq 0}$ and of $\Delta_0(n)_{n\ge 0}$ in Section~\ref{sec:block_PD}.
Observe that $$\JM_0(n+1)=\M_0(n+1)-\M_0(n)= \m_0(n+1)+\Delta_0(n+1)-\m_0(n)-\Delta_0(n).$$
Since $\JM_0(n+1)$ can only take the values $0$ and $1$, the latter relation can also be expressed using $(\m_0(n)\bmod{2})_{n\ge 0}$ and $(\Delta_0(n)\bmod{2})_{n\ge 0}$. These latter sequences are $2$-regular by Corollary~\ref{cor:2regDm_PD}. By Lemma~\ref{lem:shift}, $\JM_0(n+1)_{n\ge 0}$ is thus a combination of four $2$-regular sequences. Applying again Lemma~\ref{lem:shift}, $\JM_0(n)_{n\ge 0}$ is also $2$-regular. We can show similarly that $\jm_0(n)_{n\ge 0}$ is $2$-regular. In fact, both sequences $\JM_0(n)_{n\ge 0}$ and $\jm_0(n)_{n\ge 0}$ are $2$-automatic since they only take values $0$ and $1$. Thus, all the functions in the expression for $\mathcal{P}^{(2)}_{\mathbf{p}}(n + 1)$ are $2$-regular. 

Finally, as $\mathcal{P}^{(2)}_{\mathbf{p}}(n + 1)_{n\ge 0}$ is $2$-regular, $\mathcal{P}^{(2)}_{\mathbf{p}}(n)_{n\ge 0}$ is $2$-regular by Lemma~\ref{lem:shift}.
\end{proof}


\section{Abelian complexity of $\blo(\mathbf{t},2)$}\label{sec:block_TM}

In this section, we turn our attention to the Thue--Morse word $\mathbf{t}$.
Let $\mathbf{y}$ denote $$\blo(\mathbf{t},2)=132120132012132120121320\cdots,$$ the $2$-block coding of $\mathbf{t}$ introduced in Example~\ref{exa:ct2}.
Recall that $\mathbf{y}$ is a fixed point of the morphism $\nu$ defined by $\nu:0\mapsto 12,1\mapsto 13$, $2\mapsto 20,3 \mapsto 21$.
The approach here is similar to that of the period-doubling word: we consider in this section the abelian complexity of $\mathbf{y}$, and then we compare $\mathcal{P}^{(1)}_{\mathbf{y}}(n)$ with $\mathcal{P}^{(2)}_{\mathbf{t}}(n)$ in Section~\ref{sec:2ab_comp_TM}.

Our study of the period-doubling word in Sections~\ref{sec:block_PD} and \ref{sec:2ab_comp_PD} made substantial use of counting $0$'s in factors of $\mathbf x$.
Alternatively, we could have counted the total number of $1$'s and $2$'s in factors of $\mathbf x$, since this is equivalent information and since the letters $1$ and $2$ alternate in $\mathbf x$.

For the Thue--Morse word, the appropriate statistic for factors of $\mathbf y$ is the total number of $1$'s and $2$'s (or, equivalently, the total number of $0$'s and $3$'s).
We will show in Lemma~\ref{lem:balancedTM} that the letters $1$ and $2$ alternate in $\mathbf y$.
Therefore, for $n\in\mathbb{N}$ we set
\begin{align*}
	\M_{12}(n)&:=\max\{|u|_1+|u|_2 \, : \, u \text{ is a factor of } \mathbf{y}\text{ with }|u|=n\}, \\
	\m_{12}(n)&:=\min\{|u|_1+|u|_2 \, : \, u \text{ is a factor of } \mathbf{y}\text{ with }|u|=n\}, \\
	\Delta_{12}(n)&:=\M_{12}(n)-\m_{12}(n).
\end{align*}

\begin{remark}
    Note that $g(\mathbf{y})$ is exactly the period-doubling word $\mathbf{p}$, where $g$ is the coding defined by $g(0)=1$, $g(1)=0$, $g(2)=0$ and $g(3)=1$. In particular, $\Delta_{12}(n) + 1$ is the abelian complexity function of the period-doubling word. This function was also studied in \cite{Blanchet, Karhumaki--Saarela--Zamboni arxiv}.  Here we obtain relations of the same type as the relations in Theorem~\ref{thm:reflection_recurrence}.
\end{remark}

 The fact that $\mathcal{P}^{(1)}_{\mathbf{y}}(n)_{n \geq 0}$ is $2$-regular will follow from the next statement. 

\begin{proposition}\label{prop:deltatoab_TM}
  Let $n\in \mathbb N$. We have
  \begin{equation}
      \label{eq:conddef}
 \mathcal{P}^{(1)}_{\mathbf{y}}(n)=
  \begin{cases}
    2\Delta_{12}(n)+2 & \text{if }n \text{ is odd} \\
    \frac{5}{2}\Delta_{12}(n)+\frac{5}{2} & \text{if }n \text{ and } \Delta_{12}(n)+1 \text{ are even} \\
    \frac{5}{2}\Delta_{12}(n)+4 & \text{if $n$, $\Delta_{12}(n)$ and $\m_{12}(n)+1$ are even}\\
    \frac{5}{2}\Delta_{12}(n)+1 & \text{if $n$, $\Delta_{12}(n)$ and $\m_{12}(n)$ are even}.\\
   \end{cases}
  \end{equation}
\end{proposition}

To be able to apply the composition result given by Lemma~\ref{lem:compo} to the expression of $\mathcal{P}^{(1)}_{\mathbf{y}}$ derived in Proposition~\ref{prop:deltatoab_TM}, we have therefore to prove that 
\begin{itemize}
  \item the sequence $\Delta_{12}(n)_{n\ge 0}$ is $2$-regular and
  \item the predicates occurring in \eqref{eq:conddef} are $2$-automatic.
\end{itemize}

Section~\ref{sec:41} is dedicated to the proof of Proposition~\ref{prop:deltatoab_TM}. In Section~\ref{sec:42}, we give a proof of the two previous items. In particular, we show that $\Delta_{12}(n)_{n\ge 0}$ satisfies a reflection symmetry. This permits us to express recurrence relations for $\mathcal{P}^{(1)}_{\mathbf{y}}$ at the end of Section~\ref{sec:42}.


\subsection{Proof of Proposition~\ref{prop:deltatoab_TM}}\label{sec:41}

We first need three technical lemmas about factors of $\mathbf{y}=\blo(\mathbf{t},2)$. 

\begin{lemma}\label{lem:2facTM}
The set of factors of $\mathbf{y}$ of length $2$ is $\fac_{\mathbf{y}}(2)=\{01,12,13,20,21,32\}$.
\end{lemma}

\begin{proof}
It is easy to check that these six words are factors. To prove that they are the only ones, it is enough to check that for any element $u$ in $\{01,12,13,20,21,32\}$ the three factors of length $2$ of $\nu(u)$ are still in $\{01,12,13,20,21,32\}$.
\end{proof}

The following lemma has already been observed in \cite[Lemma~10]{Karhumaki--Saarela--Zamboni arxiv}.

\begin{lemma}\label{lem:balancedTM} 
If $w$ is a factor of $\mathbf{y}$, then $\big||w|_1-|w|_2\big|\le 1$ and $\big||w|_0-|w|_3\big|\le 1$.
In particular, the letters $1$ and $2$ (respectively $0$ and $3$) alternate in $\mathbf y$.
\end{lemma}

\begin{proof}
First note that if for all factors of a word $u$, the numbers of two letters $x$ and $y$ differ by at most $1$, then $x$ and $y$ alternate in $u$. Furthermore, if the first or the last occurrence of one of these letters is $x$, then $|u|_x\geq |u|_y$. If both the first and the last occurrences are $x$, then $|u|_x=|u|_y+1$.

We prove the result by induction on the length $\ell$ of the factor. The result is true for factors of length $\ell=1$. Let $w$ be a factor of length $\ell>1$ and assume the result holds for factors of length smaller than $\ell$. If $w$ can be de-substituted as $w=\nu(w')$, we have
\begin{align*}
|w|_0&=|w'|_2,\\
|w|_1&=|w'|_0+|w'|_1+|w'|_3,\\
|w|_2&=|w'|_0+|w'|_2+|w'|_3,\\
|w|_3&=|w'|_1.
\end{align*}
Using the induction hypothesis, we have $$\big||w|_1-|w|_2\big|=\big||w|_0-|w|_3\big|= \big||w'|_1-|w'|_2\big| \le 1.$$

If $w$ cannot be de-substituted and has odd length, we have $$w \in \left\{1^{-1}\nu(w'), \, 2^{-1}\nu(w'), \, \nu(w')1, \, \nu(w')2\right\}$$ for some factor $w'$ with $|w'|<\ell$. Assume that $w=1^{-1}\nu(w')$. Then as before $\big||w|_0-|w|_3\big|= \big||w'|_1-|w'|_2\big| \le 1$. 
For the numbers of $1$ and $2$, $w'$ starts with $0$ or $1$. Since by Lemma~\ref{lem:2facTM} a $0$ is always followed by a $1$, $w'$ starts either with $01$ or with $1$. In both cases, since $1$ and $2$ alternate, we have $|w'|_1\geq |w'|_2$ and thus
$$\big||w|_1-|w|_2\big|= \big||w'|_1-|w'|_2-1\big| \le 1.$$
The same reasoning can be done for $w=2^{-1}\nu(w')$.
If $w=\nu(w')1$, then we clearly have $\big||w|_0-|w|_3\big|\le 1$ using the result on $\nu(w')$.
By Lemma~\ref{lem:2facTM}, the factor $\nu(w')$ must end either with $0$ or $2$. So $w'$ ends with $0$ or $2$ as well. Since a $0$ is always preceded by a $2$, we necessarily have $|w'|_2\geq |w'|_1$ and
 $$\big||w|_1-|w|_2\big|= \big||w'|_1-|w'|_2+1\big| \le 1.$$
The same reasoning applies to $w=\nu(w')2$.

If $w$ cannot be de-substituted and has even length, then we have 
$$w \in \left\{1^{-1}\nu(w')1, \, 1^{-1}\nu(w')2, \, 2^{-1}\nu(w')1, \, 2^{-1}\nu(w')2\right\}$$
for some factor $w'$ with $|w'|<\ell$. If the same letter is removed and added to $\nu(w')$, then the result is clearly true. Otherwise, assume that $w=1^{-1}\nu(w')2$ (the same reasoning holds for the last case). It is clear that $\big||w|_0-|w|_3\big|\le 1$ using the result on $\nu(w')$. For the numbers of $1$ and $2$, as before, $w'$ starts with $01$ or $1$ and ends with $13$ or $1$. Hence we have $|w'|_1=|w'|_2+1$ and then 
\[
	\big||w|_1-|w|_2\big|= \big||w'|_1-|w'|_2-2\big| \le 1. \qedhere
\]
\end{proof}

\begin{lemma}\label{lem:reversalTM}
Let $\tau,\tau'$ be the morphisms respectively defined by 
\[
\tau:\left\{\begin{array}{c}
0\mapsto0\\
1\mapsto2\\
2\mapsto1\\
3\mapsto3\\
\end{array}\right.
\qquad \text{ and } \qquad
\tau':\left\{\begin{array}{c}
0\mapsto3\\
1\mapsto1\\
2\mapsto2\\
3\mapsto0\\
\end{array}\right. .
\]
If $w$ is a factor of $\mathbf{y}$, then $\tau'(w)^\mathrm{R}$, $\tau(w)^\mathrm{R}$ and $\tau'(\tau(w))$ are also factors of $\mathbf{y}$.
\end{lemma}

\begin{proof}
We prove the lemma for $\tau'(w)^\mathrm{R}$ and $\tau(w)^\mathrm{R}$ since $\tau'(\tau(w)) =\tau'(\tau(w)^\mathrm{R})^\mathrm{R}$.

We first prove by induction that for any factor $u$ starting with the letter $x$ and ending with the letter $y$, 
\begin{equation}\label{eq:rectau}
\tau'(\nu(u))^\mathrm{R} = a^{-1}\nu(\tau(u)^\mathrm{R})b
\end{equation}
where $a=1$ (respectively $a=2$, $b=1$, $b=2$) if and only if $y\in\{0,2\}$ (resp.\ $y\in \{1,3\}$, $x\in \{0,1\}$, $x\in \{2,3\}$).
Note that $a^{-1}\nu(\tau(u)^\mathrm{R})b$ is well defined. Indeed, if $y\in\{0,2\}$, then $\tau(u)^\mathrm{R}$ starts with $0$ or $1$ and thus $\nu(\tau(u)^\mathrm{R})$ starts with $a=1$. The same holds with $y\in \{1,3\}$.

The relation~\eqref{eq:rectau} is true for $u$ of length $1$. We have for example 
$$\tau'(\nu(0))^\mathrm{R} = 21 =1^{-1}\nu(0)1 =1^{-1}\nu(\tau(0)^\mathrm{R})1$$ and $$\tau'(\nu(1))^\mathrm{R} = 01 =  2^{-1}\nu(2)1=2^{-1}\nu(\tau(1)^\mathrm{R})1.$$

Let $u=u'yx$ be a factor with at least two letters $x$ and $y$.
Assume the conclusion holds for words of length at most $|u| - 1$.
By the induction hypothesis, we have $\tau'(\nu(u'y))^\mathrm{R}=a^{-1}\nu(\tau(u'y)^\mathrm{R})b$ and  $\tau'(\nu(x))^\mathrm{R}=c^{-1}\nu(\tau(x)^\mathrm{R})d$ with appropriate $a,b,c,d$. Since $yx$ is a factor, one can check using Lemma~\ref{lem:2facTM} that $a=d$. Indeed, if  $y\in \{0,2\}$, then $x\in\{0,1\}$. So $a=1$ and $d=1$. Similarly, if $y\in \{1,3\}$, then $x\in \{2,3\}$. Hence, $a=2$ and $d=2$. Thus, we have 
\begin{align*}
\tau'(\nu(u))^\mathrm{R}  &= \tau'(\nu(u'yx)^\mathrm{R})\\
        &= \tau'(\nu(x))^\mathrm{R}\tau'(\nu(u'y))^\mathrm{R}\\
        &= c^{-1}\nu(\tau(x)^\mathrm{R})da^{-1}\nu(\tau(u'y)^\mathrm{R})b\\
        &= c^{-1}\nu(\tau(u'yx)^\mathrm{R})b\\ 
        &= c^{-1}\nu(\tau(u)^\mathrm{R})b.       
\end{align*}

We can similarly prove by induction that for any factor $u$ starting with the letter $x$ and ending with the letter $y$, 
\begin{equation*}
\tau(\nu(u))^\mathrm{R} = a^{-1}\nu(\tau'(u)^\mathrm{R})b
\end{equation*}
where $a=1$ (respectively $a=2$, $b=1$, $b=2$) if and only if $y\in\{1,3\}$ (resp.\ $y\in \{0,2\}$, $x\in \{2,3\}$, $x\in \{0,1\}$).

We now prove the lemma (for $\tau$ and $\tau'$ together) by induction on the length of $w$. One can check by hand that the lemma is true for $w$ of length at most $4$. Assume the lemma is true for any factor of length at most $n\geq 4$, and let $w$ be a factor of length $n+1$.
There exist some factors $s$, $t$ and $v$ such that $swt=\nu(v)$, $0\leq |t|\leq 1$ and $1\leq |s| \leq 2$. 
Then we have $|v|\leq \frac{n+4}{2} \leq n$.
By the induction hypothesis, $\tau(v)^\mathrm{R}$ is a factor of $\mathbf{y}$. Hence $\nu(\tau(v)^\mathrm{R})$ is also a factor of $\mathbf{y}$. Using the previous result, $\tau'(\nu(v))^\mathrm{R}=a^{-1}\nu(\tau(v)^\mathrm{R})b$ for some letters $a$ and $b$. But we also have $\tau'(\nu(v))^\mathrm{R}=\tau'(t)^\mathrm{R}\tau'(w)^\mathrm{R}\tau'(s)^\mathrm{R}$ and since $s$ has at least one letter, $\tau'(w)^\mathrm{R}$ is a factor of $\nu(\tau(v)^\mathrm{R})$. Hence it is a factor of $\mathbf{y}$. We do the same proof for $\tau(w)^\mathrm{R}$.
\end{proof}

We are now ready to prove the relationship between $\mathcal{P}^{(1)}_{\mathbf{y}}(n)$ and $\Delta_{12}(n)$.

\begin{proof}[Proof of Proposition~\ref{prop:deltatoab_TM}]
Let $u$ be a factor of length $n$ of $\mathbf{y}$. Let $n_{12}=|u|_1+|u|_2$ and $n_{03}=|u|_0+|u|_3$. 

Assume first that $n$ is odd. 
If $n_{12}$ is even, then there are the same number of $1$'s and $2$'s in $u$ by Lemma~\ref{lem:balancedTM}.
Since $n_{13}$ is odd, if $|u|_0=|u|_3+1$ (resp.\ $|u|_3=|u|_0+1$), then $\tau'(u)^\mathrm{R}$ is a factor by Lemma~\ref{lem:reversalTM} and $|\tau'(u)^\mathrm{R}|_3=|\tau'(u)^\mathrm{R}|_0+1$ (resp.\ $|\tau'(u)^\mathrm{R}|_0=|\tau'(u)^\mathrm{R}|_3+1$).  In either case, $\tau'(u)^\mathrm{R}$ still has $n_{12}$ ones and twos. Hence there are exactly two abelian equivalence classes for fixed $n$ odd and $n_{12}$ even.
We can do the same reasoning if $n_{12}$ is odd. Finally, there are $\Delta_{12}(n)+1$ possible values for $n_{12}$ and thus $2(\Delta_{12}(n)+1)$ abelian equivalence classes for a fixed odd $n$.

Assume now that $n$ is even. If both $n_{12}$ and $n_{03}$ are even, then $u$ necessarily has the same number of $1$'s as $2$'s and the same number of $0$'s as $3$'s, and thus there is only one abelian equivalence class.
Hence assume that $n_{12}$ and $n_{03}$ are odd. We have $(|u|_0-|u|_3,|u|_1-|u|_2) \in \{-1,1\}^2$. By Lemma~\ref{lem:reversalTM}, the four factors  $u$, $\tau'(u)^\mathrm{R}$, $\tau(u)^\mathrm{R}$ and $\tau'(\tau(u))$ realize the four possibilities for $(|u|_0-|u|_3,|u|_1-|u|_2)$.
Hence if $n_{12}$ and $n_{03}$ are both odd, there are four abelian equivalence classes.

Now, we just have to count pairs $(n,n_{12})$ with $n$ and $n_{12}$ even.
If $\Delta_{12}(n)$ is odd, there are exactly $(\Delta_{12}(n)+1)/2$ such pairs. So there are 
$$1\cdot(\Delta_{12}(n)+1)/2+4\cdot(\Delta_{12}(n)+1)/2=\frac{5}{2}(\Delta_{12}(n)+1)$$ abelian classes for this value of $n$.
If $\Delta_{12}(n)$ is even and $\m_{12}(n)$ is odd, there are exactly $\Delta_{12}(n)/2$ even values for $n_{12}$, and so there are $$1\cdot \Delta_{12}(n)/2+4\cdot(\Delta_{12}(n)/2+1)=\frac{5}{2}\Delta_{12}(n)+4$$ abelian classes.
Finally, if $\Delta_{12}(n)$ is even and $\m_{12}(n)$ is even, there are $\Delta_{12}(n)/2+1$ even values for $n_{12}$, and so there are $$1\cdot (\Delta_{12}(n)/2+1)+4\cdot\Delta_{12}(n)/2=\frac{5}{2}\Delta_{12}(n)+1$$ abelian classes.
\end{proof}


\subsection{$\Delta_{12}(n)_{n\ge0}$ is $2$-regular, $(\m_{12}(n)\bmod{2})_{n\ge0}$ is $2$-automatic}\label{sec:42}

In this section, we prove the following result.

\begin{proposition}\label{prop:recdelta_TM}
Let $\ell\geq 1$ and $r$ such that $0\leq r < 2^{\ell}$. We have
  $$\Delta_{12}(2^\ell+r)=
  \begin{cases}
    \Delta_{12}(r)+1& \text{if } r\leq 2^{\ell-1} \\
    \Delta_{12}(2^{\ell+1}-r) & \text{if } r>2^{\ell-1}. \\
   \end{cases}
$$
Moreover,
$$\m_{12}(2^\ell+r)\equiv
  \begin{cases}
    \m_{12}(r)+\ell \pmod{2} & \text{if } r\leq 2^{\ell-1} \\
    \m_{12}(2^{\ell+1}-r)+\Delta_{12}(2^{\ell+1}-r) \pmod{2} & \text{if } r > 2^{\ell-1}. \\
   \end{cases}
$$
\end{proposition}

Note that those latter relations have a form similar to (but slightly different from) the assumptions of Theorem~\ref{thm:reflection_recurrence}. Before giving the proof, we prove a corollary. The $2$-regularity of $\mathcal{P}_{\mathbf{y}}^{(1)}(n)_{n\ge0}$ follows from Proposition~\ref{prop:deltatoab_TM} and Corollary~\ref{cor:2regDm_TM}.

\begin{corollary}\label{cor:2regDm_TM}
The following statements are true.
\begin{itemize}
	\item The sequence $\Delta_{12}(n)_{n\ge0}$ is $2$-regular.
	\item The sequence $(\Delta_{12}(n) \bmod 2)_{n\ge0}$ is $2$-automatic.
	\item The sequence $(\m_{12}(n)\bmod{2})_{n\ge0}$ is $2$-automatic.
\end{itemize}
\end{corollary}

\begin{proof}
The first assertion is a direct consequence of Proposition~\ref{prop:recdelta_TM} and Theorem~\ref{thm:reflection_recurrence}. The second assertion follows from Lemma~\ref{lem:regmod}. 

To prove the last assertion, we prove by induction that, modulo $2$,
\[
\m_{12}(16n+i) \equiv  \begin{cases}
     \m_{12}(4n) & \text{if } i=0\\
 \m_{12}(4n+1) &\text{if } i \in \{1,4,5\}\\
 \m_{12}(4n+1)+1& \text{if } i \in \{2,3\}\\
 \m_{12}(4n+2) &\text{if } i \in \{6,8,9\}\\
 \m_{12}(4n+2)+1& \text{if } i \in \{7,10\}\\
 \m_{12}(4n+3) &\text{if } i \in \{12,13,15\}\\
 \m_{12}(4n+3)+1& \text{if } i \in \{11,14\}
\end{cases}
\]
and
\[
\Delta_{12}(16n+i)\equiv 
\begin{cases}\Delta_{12}(4n)& \text{if } i=0\\
 \Delta_{12}(4n+1)& \text{if } i \in \{1,2,4\}\\
 \Delta_{12}(4n+1)+1 &\text{if } i \in \{3,5\}\\
 \Delta_{12}(4n+2)& \text{if } i=8\\
 \Delta_{12}(4n+2)+1& \text{if } i \in \{6,7,9,10\}\\
 \Delta_{12}(4n+3)& \text{if } i \in \{12,14,15\}\\
 \Delta_{12}(4n+3)+1 &\text{if } i \in \{11,13\}.
\end{cases}
\]
The relations are true for $n=0$. Let $n>0$ and assume they are true for $n'<n$. We can write $n=2^\ell+r$ with $\ell\geq 0$ and $0\leq r < 2^\ell$. Let $i\in\{0,\ldots,15\}$. We consider two cases.

Assume first that $r<2^{\ell-1}$. We have $16n+i=2^{\ell+4}+16r+i$ and $16r+i<2^{\ell+3}$.
\begin{align*}
\m_{12}(16n+i)&\equiv \m_{12}(16r+i) +\ell+4 \tag{Proposition~\ref{prop:recdelta_TM}}\\
& \equiv \m_{12}(4r+j)+\delta+\ell+4\tag{induction}\\
&\equiv \m_{12}(2^{\ell+2}+4r+j)+\delta\tag{Proposition~\ref{prop:recdelta_TM}}\\
&\equiv \m_{12}(4n+j)+\delta \pmod{2}
\end{align*}
for some $j\in\{0,\ldots,3\}$ and $\delta\in\{0,1\}$ according to the relations. A similar reasoning holds for the $\Delta_{12}$ relations.

Assume now that $r\geq 2^{\ell-1}$ and $i \neq 0$. Setting $i'=16-i$ and $n'=2^{\ell+1}-r-1$, we obtain $16n'+i'=2^{\ell+5}-16r-i$. It follows that, by Proposition~\ref{prop:recdelta_TM},
\begin{align*}
\m_{12}(16n+i)&\equiv \m_{12}(2^{\ell+5}-16r-i) + \Delta_{12}(2^{\ell+5}-16r-i)\\
&\equiv\m_{12}(16n'+i')+ \Delta_{12}(16n'+i') \\
&\equiv \m_{12}(4n'+k) +\delta + \Delta_{12}(4n'+k') +\delta'\tag{induction}
\end{align*}
for some $k,k'\in\{0,\ldots,3\}$ and $\delta,\delta'\in\{0,1\}$ according to the relations.
Note that we have $k = k'$, so
\begin{align*}
\m_{12}(16n+i)&\equiv \m_{12}(4n'+k) +\delta + \Delta_{12}(4n'+k) +\delta'\\
& \equiv \m_{12}(2^{\ell+3}-(4r+4-k)) +\delta + \Delta_{12}(2^{\ell+3}-(4r+4-k)) +\delta'\\
& \equiv \m_{12}(2^{\ell+2}+(4r+4-k)) +\delta + \delta'\tag{Proposition~\ref{prop:recdelta_TM}}\\
& \equiv \m_{12}(4n+(4-k)) +\delta + \delta' \pmod{2}.
\end{align*}

Table~\ref{tab:indicesTM} gives the values of $i'$, $k$, $\delta$ and $\delta'$ for all the values of $i\ne 0$. Observe that the values of $4-k$ and $(\delta+\delta'\bmod{2})$ are the values given in the relation for $i$. To conclude the proof, consider the case $i=0$. We have
\begin{align*}
	\m_{12}(16n)
	&\equiv \m_{12}(16(2^{\ell+1}-r)) + \Delta_{12}(16 (2^{\ell+1}-r))&\text{(Proposition~\ref{prop:recdelta_TM})}\phantom{.} \\
	&\equiv \m_{12}(4(2^{\ell+1}-r))+\Delta_{12}(4(2^{\ell+1}-r))&\text{(induction)}\phantom{.} \\
	&\equiv \m_{12}(4n) \pmod{2}&\text{(Proposition~\ref{prop:recdelta_TM})}.
\end{align*}
A similar reasoning works for the $\Delta_{12}$ relations.\qedhere
\begin{table}[h!tb]
\[\begin{array}{|c||c|c|c|c|c|c|c|c|c|c|c|c|c|c|c|}
\hline
i&1&2&3&4&5&6&7&8&9&10&11&12&13&14&15\\\hline
i'&15&14&13&12&11&10&9&8&7&6&5&4&3&2&1\\\hline
k&3&3&3&3&3&2&2&2&2&2&1&1&1&1&1\\\hline
\delta&0&1&0&0&1&1&0&0&1&0&0&0&1&1&0\\\hline
\delta'&0&0&1&0&1&1&1&0&1&1&1&0&1&0&0\\\hline
\end{array}
\]
    \caption{The corresponding values of $i'=16-i$, $k$, $\delta$ and $\delta'$.}
    \label{tab:indicesTM}
\end{table}
\end{proof}

Proposition~\ref{prop:recdelta_TM} is a direct consequence of Lemmas~\ref{lem:power2TM},~\ref{lem:recminmaxTM} and~\ref{lem:recminmax2TM} given in this section.

\begin{lemma}\label{lem:power2TM}
  Let $\ell \in \mathbb N$, $\ell\geq 1$. We have $\Delta_{12}(2^\ell)=1,\ \m_{12}(2^\ell)\equiv \ell \pmod{2},$
$$\m_{12}(2^\ell)+\M_{12}(2^{\ell+1})=2^{\ell+1}\text{ and } \M_{12}(2^\ell)+\m_{12}(2^{\ell+1})=2^{\ell+1}.$$ 
\end{lemma}

\begin{proof}
Let $\ell\geq 1$, $A_\ell=\frac{2^{\ell+1}+(-1)^{\ell}}{3}$ and $B_\ell=\frac{2^{\ell+1}+2(-1)^{\ell+1}}{3}$.
The sequences 
$$(A_\ell)_{\ell\ge 1}=(1,3,5,11,21,\ldots) \text{ and } (B_\ell)_{\ell\ge 1}=(2,2,6,10,22,\ldots)$$ 
are integer sequences and both satisfy the recurrence relation $X_{\ell+1}=2^{\ell+1}-X_\ell$. Moreover we have $A_\ell=B_\ell+1$ for even $\ell$ and $B_\ell=A_\ell+1$ for odd $\ell$. Note that $|\nu^\ell(1)|_1+|\nu^\ell(1)|_2=A_\ell$ and $|\nu^\ell(0)|_1+|\nu^\ell(0)|_2=B_\ell$.

We show by induction that 
\begin{align*}
\left\{|w|_1+|w|_2 :w\text{ factor of }\mathbf{y}\text{ with }|w|=2^\ell\right\}= \{A_\ell,B_\ell\}.
\end{align*}

Note that this result will imply the lemma and that we already have $A_\ell$ and $B_\ell$ in the set.

It is easy to check the result for $\ell=1$. Assume the result is true for $\ell\geq 1$. Let $w$ be a factor of $\mathbf{y}$ of length $2^{\ell+1}$. 
If $w$ can be de-substituted, then $w=\nu(u)$ and $|w|_1+|w|_2=2|u|_0+|u|_1+|u|_2+2|u|_3$ as in the proof of Lemma~\ref{lem:balancedTM}. Hence $|w|_1+|w|_2=2|u|-(|u|_1+|u|_2)=2^{\ell+1}-(|u|_1+|u|_2)$. Using the recurrence relation for $A_\ell$ and $B_\ell$ and since $|u|_1+|u|_2\in \{A_{\ell},B_{\ell}\}$, we have $|w|_1+|w|_2 \in \{A_{\ell+1},B_{\ell+1}\}$.
If $w$ cannot be de-substituted, then we can write $w=a^{-1}\nu(u)b$ for some letters $a,b\in \{1,2\}$ and $|\nu(u)|=2^{\ell+1}$. So $|w|_1+|w|_2=|\nu(u)|_1+|\nu(u)|_2$. Since we already proved that $|\nu(u)|_1+|\nu(u)|_2 \in \{A_{\ell+1},B_{\ell+1}\}$, we are done.

To prove the second assertion of the lemma, observe that $\m_{12}(2^\ell)=A_\ell$ if $\ell$ is odd and $\m_{12}(2^\ell)=B_\ell$ if $\ell$ is even. Furthermore, $A_\ell$ is always odd whereas $B_\ell$ is always even.
\end{proof}

In order to prove Lemmas~\ref{lem:recminmaxTM} and~\ref{lem:recminmax2TM}, we first need some technical results.

\begin{lemma}\label{lem:minmaxphi}
Let $u$ be a factor of $\mathbf{y}$ of length $n$.
We have $|u|_1+|u|_2=\M_{12}(n)$ if and only if $|\nu(u)|_1+|\nu(u)|_2=\m_{12}(2n)$, and
$|u|_1+|u|_2=\m_{12}(n)$ if and only if $|\nu(u)|_1+|\nu(u)|_2=\M_{12}(2n)$.
\end{lemma}

\begin{proof}
Recall that $|\nu(u)|_1+|\nu(u)|_2=2n-(|u|_1+|u|_2)$.
Assume that $|u|_1+|u|_2=\M_{12}(n)$ and  that $|\nu(u)|_1+|\nu(u)|_2=x>\m_{12}(2n)$. Thus $x=2n-\M_{12}(n)$. 
There exists a factor $w$ of length $2n$ with $x-1$ ones and twos. We can assume that $w$ can be de-substituted. Otherwise, we can write $w$ as $w=a^{-1}\nu(v)b$ for some $a,b\in \{1,2\}$. Thus $\nu(v)$ has the same length as $w$ and the same number of $1$'s and $2$'s. So we can assume $w=\nu(v)$. Then $|v|_1+|v|_2=2n-(x-1)=\M_{12}(n)+1$, a contradiction.

For the other direction, assume that $|u|_1+|u|_2=x<\M_{12}(n)$ and  that $|\nu(u)|_1+|\nu(u)|_2=\m_{12}(2n)$. Thus $x=n-\m_{12}(n)$. As before, there exists a factor $v$ of length $n$ with $x+1$ ones and twos. Then $\nu(v)$ has $\m_{12}(n)-1$ ones and twos, a contradiction.  

The second part of the lemma is similar.
\end{proof}

\begin{lemma}\label{lem:minmaxnodd}
Let $n$ be an odd integer. Then we have
\begin{align*}
	\m_{12}(n)&=\m_{12}(n+1)-1, \\
	\M_{12}(n)&=\M_{12}(n-1)+1.
\end{align*}
\end{lemma}

\begin{proof}
Let $u$ be a factor of even length $n+1$ minimizing the number of $1$'s and $2$'s. Then either $u$ starts with $1$ or $2$, or ends with $1$ or $2$. Indeed, if $u$ can be de-substituted, then it starts with $1$ or $2$. Otherwise, its last letter is the beginning of an image of $\nu$ and thus is $1$ or $2$.
Removing this letter, we get a word of length $n$ with $\m_{12}(n+1)-1$ ones and twos. Since the function $\m_{12}$ increases by $0$ or $1$ from $n$ to $n+1$, we have $\m_{12}(n)=\m_{12}(n+1)-1$.

For the second equality, consider a factor $u$ of even length $n-1$ with $\M_{12}(n-1)$ ones and twos. There exist two letters $a$ and $b$ such that $aub$ is a factor. Then, as before, since $aub$ has even length, $a$ or $b$ must be a $1$ or a $2$. Then $au$ or $ub$ is a factor of length $n$ with $\M_{12}(n-1)+1$ ones and twos and we conclude as before.
\end{proof}

\begin{lemma}\label{lem:recminmaxTM}
If $\ell\geq 1$ and $0\leq r\leq 2^{\ell-1}$, then 
\begin{align*}
\M_{12}(2^\ell+r)&=\M_{12}(2^\ell)+\M_{12}(r)\\
\m_{12}(2^\ell+r)&=\m_{12}(2^\ell)+\m_{12}(r).
\end{align*}
\end{lemma}

\begin{proof}
We prove the two results together by induction on $\ell$. One checks the case $\ell=1$. Let $\ell>1$ and assume the result is true for $\ell-1$. Let $r$ such that $0\leq r\leq 2^{\ell-1}$.

Assume first that $r$ is even. By the induction hypothesis, there exists a factor $u$ of length $2^{\ell-1}+r/2$ such that 
$$|u|_1+|u|_2=\m_{12}(2^{\ell-1}+r/2)=\m_{12}(2^{\ell-1})+\m_{12}(r/2).$$ 
We can write $u=vw$ with $v$ of length $2^{\ell-1}$ and $w$ of length $r/2$. Both the words $v$ and $w$ must minimize the number of $1$'s and $2$'s for their respective lengths. By Lemma~\ref{lem:minmaxphi}, $\nu(u)=\nu(v)\nu(w)$ maximizes the number of $1$'s and $2$'s and so do $\nu(v)$ and $\nu(w)$. Thus, $\M_{12}(2^\ell+r)=|\nu(u)|_1+|\nu(u)|_2$ and
\[
\M_{12}(2^\ell+r)=|\nu(v)|_1+|\nu(v)|_2+|\nu(w)|_1+|\nu(w)|_2=\M_{12}(2^{\ell})+\M_{12}(r).
\]
A similar proof shows that $\m_{12}(2^\ell+r)=\m_{12}(2^{\ell})+\m_{12}(r)$.

Assume now that $r$ is odd. We still have $0\leq r-1 < r+1\leq 2^{\ell-1}$. Hence we can apply the previous result to obtain  $\M_{12}(2^\ell+r-1)=\M_{12}(2^{\ell})+\M_{12}(r-1)$.
By Lemma~\ref{lem:minmaxnodd}, 
\begin{align*}
	\M_{12}(2^\ell+r)
	&=\M_{12}(2^\ell+r-1)+1\\
	&=\M_{12}(2^{\ell})+\M_{12}(r-1)+1\\
	&=\M_{12}(2^{\ell})+\M_{12}(r).
\end{align*}

For the $\m_{12}$ equality, a similar argument holds (using the previous result for $r+1$).
\end{proof}

\begin{lemma}
    If $\ell\geq 1$ and $2^{\ell-1}\leq r \leq 2^\ell$, then
    \begin{align*}
	\M_{12}(2^{\ell+1})&=\M_{12}(2^\ell+r)+\m_{12}(2^\ell-r)\\
	\m_{12}(2^{\ell+1})&=\m_{12}(2^\ell+r)+\M_{12}(2^\ell-r).
	\end{align*}
Moreover, there is a factor of length $2^{\ell+1}$ maximizing (resp.\ minimizing) the number of $1$'s and $2$'s such that the prefix of length $2^\ell+r$ also maximizes (resp.\ minimizes) the number of $1$'s and $2$'s.
\end{lemma}

\begin{proof}
We proceed by induction on $\ell$. The result is true for $\ell=1$ since the only non-trivial case is $r=1$. Then $\M_{12}(4)=\M_{12}(3)+\m_{12}(1)$ and $\m_{12}(4)=\m_{12}(3)+\M_{12}(1)$ and the factors $2120$ and $0132$ satisfy the claim.

Let $\ell>1$ and assume the result is true for $\ell-1$.
Let $r$ such that $2^{\ell-1}\leq r \leq 2^\ell$. Assume first that $r$ is even.
Then $2^{\ell-2}\leq r/2 \leq 2^{\ell-1}$. By the induction hypothesis, there is a factor $u$ of length $2^\ell$ minimizing the number of $1$'s and $2$'s such that the prefix $v$ of length $2^{\ell-1}+r/2$ minimizes the number of $1$'s and $2$'s. Thus we can write $u=vw$ and $|v|_1+|v|_2=\m_{12}(2^{\ell-1}+r/2)$ and necessarily $|w|_1+|w|_2=\M_{12}(2^{\ell-1}-r/2)$.
By Lemma~\ref{lem:minmaxphi}, $\nu(u)$ and $\nu(v)$ maximize the number of $1$'s and $2$'s  and $\nu(w)$ minimizes the number of $1$'s and $2$'s. So we can conclude the result.
A similar proof shows the other relation.
If $r$ is odd, then we still have  $2^{\ell-1}\leq r-1  \leq 2^\ell$ since $\ell>1$.
Thus we can use the previous result and together with Lemma~\ref{lem:minmaxnodd}, we have
\begin{align*}
\M_{12}(2^{\ell+1})&=\M_{12}(2^\ell+r-1)+\m_{12}(2^\ell-r+1)\\
&=\M_{12}(2^{\ell}+r)-1+\m_{12}(2^{\ell}-r)+1\\
&=\M_{12}(2^{\ell}+r)+\m_{12}(2^{\ell}-r).
\end{align*}

Similarly, using the fact that $r+1\leq 2^\ell$,
\begin{align*}
\m_{12}(2^{\ell+1})&=\m_{12}(2^\ell+r+1)+\M_{12}(2^\ell-r-1)\\
&=\m_{12}(2^{\ell}+r)+1+\M_{12}(2^{\ell}-r)-1\\
&=\m_{12}(2^{\ell}+r)+\M_{12}(2^{\ell}-r).
\end{align*}

For the construction of the factors, one can construct them using the factor $\nu(u)$ maximizing the number of $1$'s and $2$'s given for $r-1$ and the factor $\nu(u')$ minimizing the number of $1$'s and $2$'s given for $r+1$ in the previous construction. Since $r$ is odd, the letter between the prefix $\nu(v)$ of length $2^\ell+r-1$ and $2^\ell+r$ of $\nu(u)$ is  $1$ or $2$. Since the prefix of length $2^\ell+r-1$ of $\nu(u)$ maximizes the number of $1$'s and $2$'s, so does the prefix of length $2^\ell +r$ of $\nu(u)$. For $\m_{12}$, consider $\nu(u')$. There exist letters $a$ and $b$ such that $w=a^{-1}\nu(u')b$ is still a factor. We must have $a,b\in \{1,2\}$. Then the prefix of length $2^\ell+r$ of $w$ minimizes the number of $1$'s and $2$'s.
\end{proof}

The previous lemma permits us to reformulate some relations between the two sequences $\M_{12}(n)_{n\ge 0}$ and $\m_{12}(n)_{n\ge 0}$.

\begin{lemma}\label{lem:recminmax2TM}
If $\ell\geq 1$ and $2^{\ell-1}\leq r \leq 2^\ell$, then 
\begin{align*}
	\M_{12}(2^{\ell}+r)&=2^{\ell+1}-\m_{12}(2^{\ell+1}-r)\\
	\m_{12}(2^{\ell}+r)&=2^{\ell+1}-\M_{12}(2^{\ell+1}-r).
\end{align*}
\end{lemma}

\begin{proof}
From the previous lemma, we have $$\M_{12}(2^\ell+r)=\M_{12}(2^{\ell+1})-\m_{12}(2^\ell-r).$$ 
By Lemma~\ref{lem:power2TM}, we have 
$\M_{12}(2^{\ell+1})=2^{\ell+1}-\m_{12}(2^\ell)$. Moreover, by Lemma~\ref{lem:recminmaxTM}, since $0\le 2^{\ell}-r\le 2^{\ell-1}$, we get
$$\m_{12}(2^\ell-r)=\m_{12}(2^{\ell}+2^{\ell}-r)-\m_{12}(2^\ell).$$
Similar relations hold when changing $\M_{12}$ to $\m_{12}$.
\end{proof}

The proof of Proposition~\ref{prop:recdelta_TM} about the reflection relation satisfied by $\Delta_{12}(n)$ and the recurrence relation of $\m_{12}(n)$ is now immediate. 
\begin{proof}[Proof of Proposition~\ref{prop:recdelta_TM}]
If $\ell\geq 1$ and  $0\leq r\leq 2^{\ell-1}$, then subtracting the two relations provided by Lemma~\ref{lem:recminmaxTM} gives 
$$\Delta_{12}(2^\ell+r)=\Delta_{12}(\ell)+\Delta_{12}(r)$$
and we can conclude using the first relation given in Lemma~\ref{lem:power2TM}, $\Delta_{12}(2^\ell)=1$. 
By Lemma~\ref{lem:recminmaxTM}, $\m_{12}(2^\ell+r)\equiv \m_{12}(2^\ell)+\m_{12}(r) \pmod{2}$. The expression for $\m_{12}(2^\ell+r)$ follows since $\m_{12}(2^\ell)\equiv \ell \pmod{2}$ by Lemma~\ref{lem:power2TM}.

If $\ell\geq 1$ and $2^{\ell-1} < r < 2^\ell$, then  subtracting the two relations provided by Lemma~\ref{lem:recminmax2TM} permits us to conclude the proof of the expression claimed for $\Delta_{12}(2^\ell+r)$.
Moreover, using  Lemma~\ref{lem:recminmax2TM}, we get
\begin{align*}
\m_{12}(2^\ell+r)&\equiv \M_{12}(2^{\ell+1}-r) \pmod{2}\\
&\equiv  \m_{12}(2^{\ell+1}-r)+\Delta_{12}(2^{\ell+1}-r) \pmod{2}. \qedhere
\end{align*}
\end{proof}

Using Propositions~\ref{prop:deltatoab_TM} and~\ref{prop:recdelta_TM}, we can express recurrence relations for $\mathcal{P}^{(1)}_{\mathbf{y}}$ as we did for the proof of Theorem~\ref{thm:recab_PD}.
 
\begin{theorem}\label{thm:recab_TM}
Let $\ell\geq 2$ and $r$ such that $0\leq r < 2^{\ell}$. For $r\leq 2^{\ell-1}$, we have
{\small
$$ \mathcal{P}^{(1)}_{\mathbf{y}}(2^\ell+r)=
  \begin{cases}
    \mathcal{P}^{(1)}_{\mathbf{y}}(r)+2 & \text{if $r$ is odd}\\
    \mathcal{P}^{(1)}_{\mathbf{y}}(r)+1 & \text{if $(r$, $\Delta_{12}(2^\ell+r)$ and }\m_{12}(2^\ell+r)\text{ are even})\\
    & \text{or }(  r\text{ and } \Delta_{12}(2^\ell+r)+1\text{ are even }\\
    &\text{\phantom{or }and }\m_{12}(2^\ell+r)\equiv \ell+1 \pmod{2})\\    
    \mathcal{P}^{(1)}_{\mathbf{y}}(r)+4 & \text{otherwise.} 
   \end{cases}$$}
For $ r>2^{\ell-1}$, we have $\mathcal{P}^{(1)}_{\mathbf{y}}(2^\ell+r)=\mathcal{P}^{(1)}_{\mathbf{y}}(2^{\ell+1}-r)$.
\end{theorem}


\section{$2$-abelian complexity of the Thue--Morse word}\label{sec:2ab_comp_TM}

The aim of this section is to express, in Theorem~\ref{thm:2abtoabTM}, $\mathcal{P}^{(2)}_{\mathbf{t}}(n+1)$ in terms of $\mathcal{P}^{(1)}_{\mathbf{y}}(n)$, $\Delta_{12}(n)$, $(\m_{12}(n) \bmod{2})$ and two new functions $\JM_{03}(n)$ and $\jm_{03}(n)$ that are defined analogously to $\JM_0(n)$ and $\jm_0(n)$ of Section~\ref{sec:2ab_comp_PD}.
Let
\begin{align*}
	\M_{03}(n)&:=\max\{|u|_0+|u|_3 \, : \, u \text{ is a factor of } \mathbf{y}\text{ with }|u|=n\}, \\
	\m_{03}(n)&:=\min\{|u|_0+|u|_3 \, : \, u \text{ is a factor of } \mathbf{y}\text{ with }|u|=n\},
\end{align*}
and let
\begin{align*}
	\JM_{03}(n)&:=
	\begin{cases}
		1	& \text{if $\M_{03}(n)>\M_{03}(n-1)$} \\
		0	& \text{otherwise},
	\end{cases} \\
	\jm_{03}(n)&:=
	\begin{cases}
		1	& \text{if $\m_{03}(n+1) > \m_{03}(n)$} \\
		0	& \text{otherwise}.
	\end{cases}
\end{align*}

\begin{theorem}\label{thm:2abtoabTM}
For $n$ odd, we have
\begin{multline*}
	\mathcal{P}^{(2)}_{\mathbf{t}}(n+1)-\mathcal{P}^{(1)}_{\mathbf{y}}(n) =\\
	\begin{cases}
		\Delta_{12}(n)+2-2\JM_{03}(n)-2\jm_{03}(n)	&\text{if $\m_{12}(n)$ and $\Delta_{12}(n)$ are even}\\
		\Delta_{12}(n)+1-2\JM_{03}(n)			&\text{if $\m_{12}(n)$ and $\Delta_{12}(n)+1$ are even}\\
	\Delta_{12}(n)+1-2\jm_{03}(n)				&\text{if $\m_{12}(n)$ and $\Delta_{12}(n)$ are odd}\\
	\Delta_{12}(n)						&\text{if $\m_{12}(n)+1$ and $\Delta_{12}(n)$ are even}.
	\end{cases}
\end{multline*}
For $n$ even, we have
\[
\mathcal{P}^{(2)}_{\mathbf{t}}(n+1)-\mathcal{P}^{(1)}_{\mathbf{y}}(n)=
	\begin{cases}
		\frac{1}{2}\Delta_{12}(n)+1		&\text{if }\m_{12}(n) \text{ and } \Delta_{12}(n) \text{ are even}\\
		\frac{1}{2}\Delta_{12}(n)		&\text{if }\m_{12}(n)+1 \text{ and } \Delta_{12}(n) \text{ are even}\\
		\frac{1}{2}\Delta_{12}(n)+\frac{1}{2}	&\text{if }\Delta_{12}(n) \text{ is odd}.\\
	\end{cases}
\]
\end{theorem}

As in Section~\ref{sec:2ab_comp_PD}, we study when an abelian equivalence class of $\mathbf{y}=\blo(\mathbf{t},2)$ splits into two $2$-abelian equivalence classes of $\mathbf{t}$. We have similar propositions.

\begin{proposition}
Let $u$ and $v$ be factors of $\mathbf{t}$ of length $n$. Let $u'$ and $v'$ be the $2$-block codings of $u$ and $v$. The factors $u$ and $v$ are $2$-abelian equivalent if and only if $u'$ and $v'$ (of length $n-1$) are abelian equivalent and either $u'$ and $v'$ both have first letter in $\{0,1\}$ or both have first letter in $\{2,3\}$.
\end{proposition}

Let $\mathcal X$ be an abelian equivalence class of factors of $\mathbf{y}$ of length $n$. For a letter $a$, let $n_{a}$ denote the number of $a$'s in each element of $\mathcal X$ and let $n_{12}=n_1+n_2$, $n_{03}=n_0+n_3$.

\begin{lemma}\label{lem:n12odd}
If $n_{12}$ is odd, then $\mathcal X$ leads to a unique $2$-abelian equivalence class of $\mathbf{t}$.
\end{lemma}

\begin{proof}
Assume that $n_1>n_2$ (the other case is similar). Then a word of $\mathcal X$ cannot start with $2$ since the letters $1$ and $2$ alternate in $\mathbf y$ by Lemma~\ref{lem:balancedTM}. It cannot start with $3$ neither since $n_1>n_2$ and a $3$ is always followed by $2$ by Lemma~\ref{lem:2facTM}. Hence it starts with $0$ or $1$. Thus $\mathcal X$ leads to a unique $2$-abelian equivalence class.
\end{proof}

\begin{lemma}\label{lem:neven}
If $n$ and $n_{12}$ are even, then $\mathcal X$ splits into two $2$-abelian equivalence classes of $\mathbf{t}$.
\end{lemma}

\begin{proof}
If $n$ and $n_{12}$ are even, then $n_{03}$ is also even and thus $n_1=n_2$ and $n_0=n_3$.
Let $u$ be an element of $\mathcal X$. Then $u' = \tau'(\tau(u))$ is also an element of $\mathcal X$.  Moreover, the first letter of $u$ is in $\{0,1\}$ if and only if the first letter of $u'$ is in $\{2,3\}$. Hence $\mathcal X$ splits into two $2$-abelian equivalence classes.
\end{proof}

So the last and hardest case happens when $n$ is odd and $n_{12}$ is even, i.e.,\ when $n$ and $n_{03}$ are odd. The $\JM_{03}$ and $\jm_{03}$ functions permit us to handle this case.

\begin{lemma}\label{lem:ext03}
Let $n$ and $n_{03}$ are odd. Let $a\in \{0,3\}$ (resp.\ $b\in \{0,3\}$) be the letter in majority (resp.\ in minority) in factors in $\mathcal{X}$, among $\{0,3\}$.
\begin{itemize}
\item We have $n_{03}=\M_{03}(n)$ and $\JM_{03}(n)=1$ if and only if every factor in $\mathcal X$ starts and ends with $a$.
\item We have $n_{03}=\m_{03}(n)$ and $\jm_{03}(n)=1$ if and only if every factor in $\mathcal X$ is preceded and followed by $b$.
\end{itemize}
\end{lemma}

\begin{proof}
Assume that $a=0$ and $b=3$ (the other case is symmetric). We first prove the statement for the maximum. Assume that all the factors in $\mathcal X$ start and end with $0$. If $n_{03}<\M_{03}(n)$, by continuity of the number of $0$'s and $3$'s and since $\mathbf{y}$ is uniformly recurrent, there exists a factor $yuz$ such that the factor $yu$ (resp.\ $uz$) is of length $n$ with $n_{03}$ (resp.\ $n_{03}+1$) zeros and threes. We necessarily have $z\in \{0,3\}$ and $u$ is not finishing with a letter in $\{0,3\}$. Since $yu$ has $n_{03}$ zeros and threes, $yu$ or $\tau'(yu)^{\mathrm R}$ is an element of $\mathcal X$ that is either not finishing or not starting with $0$, a contradiction. Hence we have $n_{03}=\M_{03}(n)$. Assume now that $\M_{03}(n-1)=n_{03}$. There exists a factor $u$ of even length $n-1$ with $n_{03}$ zeros. Without loss of generality, we can assume that $u$ has more $0$'s than $3$'s (otherwise one can consider $\tau'(u)^{\mathrm R}$ by Lemma~\ref{lem:reversalTM}). Since $u$ has even length, either $u$ occurs at an even index in $\mathbf{y}$ and is always followed by $1$ or $2$, or $u$ occurs at an odd index in  $\mathbf{y}$ and is always preceded by $1$ or $2$. In other words, there is a factor of the form $yu$ or $uy$ with $y\in \{1,2\}$. Then $yu$ or $uy$ is an element of $\mathcal X$ with the first or last letter different from $0$, a contradiction.

For the  other direction, assume that $n_{03}=\M_{03}(n)$ and $\JM_{03}(n)=1$. Let $u$ be a factor in $\mathcal X$. If $u=xu'$ or $u=u'x$ with $x\neq 0$, then $u'$ has length $n-1$ and $n_{03}$ zeros and threes. Thus $\JM_{03}(n)=0$, a contradiction.

The second statement is proved in the same way. Assume that all the factors in $\mathcal X$ are preceded and followed by $3$. If $n_{03}>\m_{03}(n)$, by continuity of the number of $0$'s and $3$'s and since $\mathbf{y}$ is uniformly recurrent, there exists a factor $yuz$ such that the factor $yu$ (resp.\ $uz$) is of length $n$ with $n_{03}$ (resp.\ $n_{03}-1$) zeros and threes. We necessarily have $z\in \{1,2\}$. Then as before $yu$ or $\tau'(yu)^{\mathrm R}$ is and element of $\mathcal X$ that is either not always followed  or not always preceded by $3$, a contradiction. Hence we have $n_{03}=\m_{03}(n)$. Assume now that $\m_{03}(n+1)=n_{03}$. There exists a factor $u$ of even length $n+1$ with $n_{03}$ zeros. Without loss of generality, we can assume that $u$ has more $0$'s than $3$'s (otherwise one can consider $\tau'(u)^{\mathrm R}$ by Lemma~\ref{lem:reversalTM}). Since $u$ has even length, either $u$ occurs at an even index and starts with $1$ or $2$ or $u$ occurs at an odd index and ends with $1$ or $2$. In other words, $u=yu'$ or $u=u'y$ with $y\in \{1,2\}$ and $u'$ is an element of $\mathcal X$ preceded or followed by a letter different from $3$, a contradiction.

For the  other direction, assume that $n_{03}=\m_{03}(n)$ and $\jm_{03}(n)=1$. Let $u$ be a factor in $\mathcal X$. If $u'=ux$ or $u'=xu$ is a factor with $x\in\{1,2\}$, then $u'$ has length $n+1$ and $n_{03}$ zeros and threes. So $\jm_{03}(n)=0$, which is a contradiction. Observe also that it is impossible to have $0u$ or $u0$ as factors of $\mathbf{y}$ since $|u|_0>|u|_3$ by assumption and the letters $0$ and $3$ alternate in $\mathbf{y}$ by Lemma~\ref{lem:balancedTM}. The conclusion is immediate.
\end{proof}

\begin{lemma}\label{lem:nodd}
If $n$ is odd and $n_{12}$ is even, then $\mathcal X$ leads to only one $2$-abelian equivalence class of $\mathbf{t}$ if and only if $n_{03}=\m_{03}(n)$ and $\jm_{03}(n)=1$, or $n_{03}=\M_{03}(n)$ and $\JM_{03}(n)=1$.
Otherwise, $\mathcal X$ splits into two classes.
\end{lemma}

\begin{proof}
If $n$ is odd and $n_{12}$ is even, then $n_{03}$ is even. Assume that $n_0>n_3$ (the other case is symmetric).
If $n_{03}=\m_{03}(n)$ and $\jm_{03}(n)=1$ then, by Lemma~\ref{lem:ext03}, all the factors in $\mathcal X$ start with $0$, and so $\mathcal X$ leads to only one class. 
If $n_{03}=\M_{03}(n)$ and $\JM_{03}(n)=1$, then all the factors in $\mathcal X$ are preceded and followed by $3$. In particular, they all start with $2$ and again $\mathcal X$ leads to only one class. 

For the other direction, suppose that $\mathcal X$ leads to only one class. All the factors in $\mathcal X$ must start either with a letter in $\{0,1\}$ or with a letter in $\{2,3\}$. Assume first that all the elements of $\mathcal X$ start with $0$ or $1$. Let $u$ be a factor in $\mathcal{X}$. If the first letter of $u$ is $1$, it must start with $120$ since $u$ has more $0$'s than $3$'s. Thus $u$ is always preceded by $2$. It cannot end with $1$ (since $n_1=n_2$). So it must end with $0$ or $2$. If $u=120u'2$, then $2120u'$ is an element of $\mathcal X$ starting with $2$, which is a contradiction. If $u=120u'0$ then $u1$ is a factor of $\mathbf y$. So $20u'01$ is an element of $\mathcal X$ starting with $2$, a contradiction.
Hence $u$ cannot start with $1$ and thus starts with $0$. Observe that, if $u$ does not end with $0$, then $\tau(u)^\mathrm{R}$ is still an element of $\mathcal X$ by Lemma~\ref{lem:reversalTM} and $\tau(u)^\mathrm{R}$ does not start with $0$, a contradiction. Hence all the factors in $\mathcal X$ start and end with $0$. By Lemma~\ref{lem:ext03}, we have $n_{03}=\M_{03}(n)$ and $\JM_{03}(n)=1$.

Assume now that all the elements of $\mathcal X$ start with $2$ or $3$. Since $n_0>n_3$, they all start with $2$. Moreover, as $n_1=n_2$, they must end with $0$ or $1$. If $u\in\mathcal{X}$ ends with $0$, then $\tau'(u)^{\mathrm R}\in\mathcal{X}$ starts with $3$ by Lemma~\ref{lem:reversalTM}, a contradiction. So all factors in $\mathcal{X}$ end with $1$. Let $u=2u'1$ be an element of $\mathcal X$. By Lemma~\ref{lem:2facTM}, the only possible extensions of $u$ as a factor of length $n+1$ of $\mathbf{y}$ are $1u$, $3u$, $u2$ and $u3$. If $1u$ is a factor of $\mathbf{y}$, then $12u'\in\mathcal X$ starts with $1$, which is a contradiction. If $u2$ is factor of $\mathbf{y}$, then $\tau(u'12)^\mathrm{R}\in\mathcal X$ starts with $1$, a contradiction. Hence all the factors in $\mathcal X$ are preceded and followed by $3$ in $\mathbf{y}$. By Lemma~\ref{lem:ext03}, this means that $n_{03}=\m_{03}(n)$ and $\jm_{03}(n)=1$.
\end{proof}

We are now ready to prove Theorem~\ref{thm:2abtoabTM}.

\begin{proof}[Proof of Theorem~\ref{thm:2abtoabTM}]

The difference between $\mathcal{P}^{(2)}_{\mathbf{t}}(n+1)$ and $\mathcal{P}^{(1)}_{\mathbf{y}}(n)$ is the number of abelian equivalence classes of factors of length $n$ of $\mathbf{y}$ that split into two $2$-abelian equivalence classes of factors of length $n+1$ of $\mathbf{t}$.

For even $n$, by Lemmas~\ref{lem:n12odd} and \ref{lem:neven}, it happens when $n_{12}$ is even. The number of even values of $n_{12}\in \{\m_{12}(n),\dots,\M_{12}(n)\}$ is 
\[
	\begin{cases}
		\frac{1}{2}\Delta_{12}(n)+1		&\text{if }\m_{12}(n) \text{ and } \Delta_{12}(n) \text{ are even}\\
		\frac{1}{2}\Delta_{12}(n)		&\text{if }\m_{12}(n)+1 \text{ and }  \Delta_{12}(n) \text{ are even} \\
		\frac{1}{2}\Delta_{12}(n)+\frac{1}{2}	&\text{if }\Delta_{12}(n) \text{ is odd,}
	\end{cases}
\]
which leads to the result.

For odd $n$, by Lemmas~\ref{lem:n12odd} and \ref{lem:nodd}, it happens when $n_{12}$ is even, except if $n_{03}=\m_{03}(n)$ and $\jm_{03}(n)=1$ or $n_{03}=\M_{03}(n)$ and $\JM_{03}(n)=1$. The number of such cases is
\[
	\begin{cases}
\frac{\Delta_{12}(n)}{2}+1-\JM_{03}(n)-\jm_{03}(n) &\text{if }\m_{12}(n) \text{ and } \Delta_{12}(n) \text{ are even}\\
\frac{\Delta_{12}(n)+1}{2}-\JM_{03}(n) &\text{if }\m_{12}(n) \text{ and } \Delta_{12}(n)+1 \text{ are even}\\
\frac{\Delta_{12}(n)+1}{2}-\jm_{03}(n) &\text{if }\m_{12}(n) \text{ and } \Delta_{12}(n) \text{ are odd}\\
\frac{\Delta_{12}(n)}{2} &\text{if }\m_{12}(n)+1 \text{ and } \Delta_{12}(n) \text{ are even.}
	\end{cases}
\]

Indeed, consider for example the case that $\m_{12}(n)$ and $\Delta_{12}(n)$ are even. First, there are $\frac{\Delta_{12}(n)}{2}+1$ even values of $n_{12}$. Second, since $\m_{12}(n)$ is even and $n$ is odd, we have $\M_{03}(n)=n-\m_{12}(n)$ odd. Since $\Delta_{12}(n)$ is even, $\M_{12}(n)$ is also even and $\m_{03}(n)$ is odd. 

If $n$ is such that $\jm_{03}(n)=1$ (resp.\ $\JM_{03}(n)=1$) then the case $n_{03}=\m_{03}(n)$ and  $\jm_{03}(n)=1$ (resp.\ $n_{03}=\M_{03}(n)$ and $\JM_{03}(n)=1$) indeed happens. So we have to remove $1$, i.e., $\jm_{03}(n)$ or $\JM_{03}(n)$ for each case.

As another example, consider the case that $\m_{12}(n)$ and $\Delta_{12}(n)$ are odd. Then $\M_{03}(n)$ is even and $\m_{03}(n)$ is odd. There are $\frac{\Delta_{12}(n)+1}{2}$ even values of $n_{12}$. We cannot have $n_{03}=\M_{03}(n)$ (for parity reasons) and thus we never have $n_{03}=\M_{03}(n)$ and $\JM_{03}(n)=1$. But the case $n_{03}=\m_{03}(n)$ happens and thus we have to remove one case when $\jm_{03}(n)=1$.

Finally, observe that to each pair $(n,n_{12})$, with $n$ odd and $n_{12}$ even, correspond two abelian equivalence classes of $\mathbf{y}$ (see the proof of Proposition~\ref{prop:deltatoab_TM}). Each of these classes splits into two $2$-abelian equivalence classes. Hence multiplying by $2$ the number of pairs $(n,n_{12})$, with $n$ odd and $n_{12}$ even, gives the result claimed for $n$ odd.
\end{proof}

\begin{corollary}
The sequence $\mathcal{P}_{\mathbf{t}}^{(2)}(n)_{n\ge0}$ is $2$-regular.
\end{corollary}

\begin{proof}
We can make use of Lemma~\ref{lem:compo}. Thanks to Theorem~\ref{thm:2abtoabTM}, $\mathcal{P}^{(2)}_{\mathbf{t}}(n+1)$ can be expressed as a combination of $\mathcal{P}^{(1)}_{\mathbf{y}}(n)$, $\Delta_{12}(n)$, $\JM_{03}(n)$, $\jm_{03}(n)$ using the predicates $(n\bmod{2})$, $(\Delta_{12}(n)\bmod{2})$ and $(\m_{12}(n)\bmod{2})$.

The sequences $\mathcal{P}^{(1)}_{\mathbf{y}}(n)_{n \geq 0}$ and $\Delta_{12}(n)_{n\ge 0}$ are $2$-regular from Section~\ref{sec:block_TM}. Note that we have $\JM_{03}(n+1)=\m_{12}(n)-\m_{12}(n+1)+1$ and
\begin{align*} 
	\jm_{03}(n)&=\M_{12}(n)-\M_{12}(n+1)+1\\
	&=\m_{12}(n)-\m_{12}(n+1)+\Delta_{12}(n)-\Delta_{12}(n+1)+1.
\end{align*}
As $\JM_{03}(n+1)$ and $\jm_{03}(n)$ can only take the values $0$ and $1$, these relations can also be expressed using $(\m_{12}(n)\bmod{2})_{n\ge 0}$ and $(\Delta_{12}(n)\bmod{2})_{n\ge 0}$. Since these two latter sequences are $2$-regular, the sequences $(\m_{12}(n+1)\bmod{2})_{n\ge 0}$ and $(\Delta_{12}(n+1)\bmod{2})_{n\ge 0}$ are $2$-regular by Lemma~\ref{lem:shift} and so are $\JM_{03}(n+1)_{n\ge 0}$ and  $\jm_{03}(n)_{n\ge 0}$ by Lemma~\ref{lem:compo}. Thus, $\JM_{03}(n)_{n\ge 0}$ is $2$-regular by Lemma~\ref{lem:shift}.

Since all the functions (resp.\ all the predicates) occurring in the statement of Theorem~\ref{thm:2abtoabTM} are $2$-regular (resp.\ $2$-automatic), the composition given in Lemma~\ref{lem:compo} implies that the sequence $\mathcal{P}^{(2)}_{\mathbf{t}}(n+1)_{n\ge 0}$ is $2$-regular. Hence, by Lemma~\ref{lem:shift}, $\mathcal{P}^{(2)}_{\mathbf{t}}(n)_{n\ge 0}$ is $2$-regular.
\end{proof}

\section{Conclusions}\label{Conclusions}
The two examples treated in this paper, namely the $2$-abelian complexity of the period-doubling word and the Thue--Morse word, suggest that a general framework to study the $\ell$-abelian complexity of $k$-automatic sequences may exist. As an example, we consider the $3$-block coding of the period-doubling word,
$$\mathbf{z}=\blo(\mathbf{p},3)=240125252401240124\cdots.$$
The abelian complexity $\mathcal{P}^{(1)}_\mathbf{z}(n)_{n\ge 0}=(1,5,5,8,6,10,19,11,\ldots)$ seems to satisfy, for $\ell\ge 4$, the following relations (which are quite similar to what we have discussed so far)
$$
\mathcal{P}^{(1)}_\mathbf{z}(2^\ell + r) =
	\begin{cases}
		\mathcal{P}^{(1)}_\mathbf{z}(r) + 5	& \text{if $r \leq 2^{\ell-1}$ and $r$ even} \\
		\mathcal{P}^{(1)}_\mathbf{z}(r) + 7	& \text{if $r \leq 2^{\ell-1}$ and $r$ odd} \\
		\mathcal{P}^{(1)}_\mathbf{z}(2^{\ell+1}-r)	& \text{if $r > 2^{\ell-1}$}.
	\end{cases}
$$
Then, the next step would be to relate $\mathcal{P}^{(3)}_\mathbf{p}$ with $\mathcal{P}^{(1)}_\mathbf{z}$ (and try to extend the developments from Section~\ref{sec:2ab_comp_PD}).

\section*{Acknowledgments}
We thank Jeffrey Shallit for some motivating discussions we had about this problem at some early stage of development.

\end{document}